\documentclass{iopart}
\usepackage[utf8]{inputenc} 
\usepackage[colorlinks=true, citecolor=blue, filecolor=black,linkcolor=blue, urlcolor=blue, hypertexnames=false]{hyperref}
\usepackage{color}
\newcommand{\blue}[1]{#1}

\usepackage{array,multirow}
\usepackage{enumitem}

\usepackage{rotating}
\usepackage{multirow}

\usepackage{afterpage}
\usepackage{pdflscape}
\usepackage{booktabs}

\usepackage{tikz}
\usetikzlibrary{positioning,shapes,calc,3d}
\usepackage{pgfplots}

\expandafter\let\csname equation*\endcsname\relax
\expandafter\let\csname endequation*\endcsname\relax
\usepackage{amsmath}
\usepackage{amssymb}
\usepackage{amsmath}
\usepackage{amsfonts}
\usepackage{amsthm}
\usepackage{iopams}
\usepackage{float}

\usepackage{graphicx}
\usepackage{geometry}
\usepackage{subcaption}
\usepackage{algorithm}
\usepackage{algpseudocode}

\usepackage{amsopn}

\theoremstyle{definition}
\newtheorem{theorem}{Theorem}[section]
\newtheorem{proposition}[theorem]{Proposition}

\newtheorem{assumption}[theorem]{Assumption}
\newtheorem{corollary}[theorem]{Corollary}

\theoremstyle{definition}
\newtheorem{definition}[theorem]{Definition}

\theoremstyle{remark}
\newtheorem{remark}[theorem]{Remark}

\DeclareMathOperator*{\argmax}{arg\,max}

\newcommand{\im}{{\rm Im}}

\newcommand{\lkd}{\mathcal{L}}
\newcommand{\potential}{\eta}

\newcommand{\expect}{\mathbb{E}}
\newcommand{\forward}{F}
\newcommand{\linear}{\mathrm J}
\newcommand{\data}{\bs y}

\newcommand{\gradientuk}{\nabla_{u}\potential(u^\ast;\data)}

\newcommand{\hessian}{\mathrm H}

\newcommand{\normal}{\mathcal{N}}
\newcommand{\hilbert}{\mathcal{H}}

\newcommand{\dataspace}{\mathbb{R}^d}

\newcommand{\prcov}{\Gamma_{\mathrm{pr}}}
\newcommand{\obscov}{\Gamma_{\mathrm{obs}}}

\newcommand{\rand}{\xi}
\newcommand{\eigentrunc}{\varrho}

\newcommand{\R}{\mathbb{R}}

\newcommand{\cA}{\mathrm A}
\newcommand{\cB}{\mathrm B}
\newcommand{\cG}{\mathrm G}
\newcommand{\cD}{\mathrm D}
\newcommand{\cI}{\mathrm I}

\newcommand{\cU}{\mathrm U}

\newcommand{\cW}{\mathrm W}
\newcommand{\basis}{\mathrm \Psi}

\newcommand{\cT}{\mathrm T}

\newcommand{\cP}{\mathrm P}
\newcommand{\cZ}{\mathrm Z}

\newcommand{\cN}{\mathcal{N}}

\newcommand{\cJ}{\mathrm{J}}

\newcommand{\Ev}{\mathbb{E}}

\newcommand{\cO}{\mathcal{O}}

\newcommand{\Var}{\mathrm{Var}}
\newcommand{\Cov}{\mathrm{Cov}}

\renewcommand{\bs}{\boldsymbol}

\renewcommand{\bxi}{\bs\xi}

\newcommand{\bV}{\bs V}

\newcommand{\be}{\bs e}
\newcommand{\bx}{\bs x}
\newcommand{\by}{\bs y}

\newcommand{\bv}{\bs v}

\renewcommand{\br}{\bs r}
\newcommand{\bn}{\bs n}

\newcommand{\bGobs}{\Gamma_{\textnormal{obs}}}
\newcommand{\bGobsi}{\Gamma_{\textnormal{obs}}^{-1}}
\newcommand{\bGpr}{\Gamma_{\textnormal{pr}}}

\newcommand{\range}{{\rm range}}
\newcommand{\mell}{\ell\text{\,-\,}1}

\newcommand{\rvl}{ {\bs V}_{\!\!\ell} }
\newcommand{\rvml}{ {\bs V}_{\!\!\mell} }

\newcommand*\circled[1]{\tikz[baseline=(char.base)]{
            \node[shape=circle,draw,inner sep=0.5pt] (char) {#1};}}

\newcommand{\ff}{f\!f}

\begin{document}

\title[MLDILI for inverse problems]{Multilevel dimension-independent likelihood-informed MCMC for large-scale inverse problems}

\author{Tiangang Cui$^{1}$, Gianluca Detommaso$^{2}$, Robert Scheichl$^{3}$}
\address{$^1$School of Mathematics and Statistics, University of Sydney, NSW 2006, Australia}
\address{$^2$Amazon Web Services Germany GmbH,  Krausenstraße 38, 10117 Berlin, Germany}
\address{$^3$Institute for Applied Mathematics and Interdisciplinary Center for Scientific Computing, Heidelberg University, Im Neuenheimer Feld 205, 69120 Heidelberg, Germany}
\ead{\mailto{tiangang.cui@sydney.edu.au}, \mailto{detomma@amazon.de}, \mailto{r.scheichl@uni-heidelberg.de}}

\setlist[enumerate]{leftmargin=3em}


\begin{abstract}
We present a non-trivial integration of dimension-independent
likelihood-informed (DILI) MCMC (Cui, Law, Marzouk, 2016) and the
multilevel MCMC (Dodwell et al., 2015) to explore the hierarchy of
posterior  distributions.
This integration offers several advantages:
First, DILI-MCMC employs an intrinsic {\it likelihood-informed
subspace} (LIS) (Cui et al., 2014)---which involves a
number of  forward and adjoint model simulations---to design
accelerated operator-weighted proposals. 
By exploiting the multilevel structure of the discretised parameters
and discretised forward models, we design a {\it Rayleigh-Ritz
procedure} to significantly reduce the computational effort in
building the LIS and operating with DILI proposals. 
Second, the resulting DILI-MCMC can drastically improve the sampling
efficiency of MCMC at each level, and hence reduce the integration
error of the multilevel algorithm for fixed CPU time.  
Numerical results confirm the improved computational efficiency of the multilevel DILI approach.\\

\end{abstract}


\section{Introduction}

Inverse problems aim to estimate unknown parameters of mathematical models from noisy and indirect observations. 
The unknown parameters, often represented as functions, are related to the observed data through a forward model, such as a differential equation, that maps realisations of parameters to observables. 
\blue{For ill-posed inverse problems}, there may exist many feasible realisations of parameters that are consistent with the observed data, and small perturbations in the data may lead to large perturbations in unregularised parameter estimates.
The Bayesian approach \cite{IP:KaiSo_2005, IP:Stuart_2010, IP:Tarantola_2004} casts the solution of inverse problems as the posterior probability distribution of the model parameters conditioned on the data.
This offers a natural way to integrate the forward model and the data together with prior knowledge and a stochastic description of measurement and/or model errors to remove the ill-posedness and to quantify uncertainties in parameters and parameter-dependent predictions.
As a result, parameter estimations, model predictions, and associated uncertainty quantifications can be issued in the form of marginal distributions or expectations of some quantities of interest (QoI) over the posterior. 
Due to the typically high parameter dimensions and the high computational cost of the forward models, characterising the posterior and computing posterior expectations are in general computationally challenging tasks.
Integrating multilevel Markov chain Monte Carlo (MCMC) \cite {MCMC:KST_2013,MCMC:HSS_2013}, likelihood-informed parameter reduction \cite{DimRedu:Cui_etal_2014,DimRedu:Spantini_etal_2015,DimRedu:Zahm_etal_2018} and dimension-independent MCMC \cite{MCMC:BRSV_2008,MCMC:CRSW_2013,MCMC:CLM_2016,MCMC:DS_2018}, we present here an integrated framework to significantly accelerate the computation of posterior expectations for large-scale inverse problems. 

In inverse problems, unknown parameters are often cast as functions, and hence the Bayesian inference has to be carried out over typically {\it high-dimensional discretisations} of the parameters that resolve the spatial and/or temporal variability of the underlying problem sufficiently. 
Examples are the permeability field of a porous medium \cite{IP:CFO_2011, MCMC:KST_2013, IP:HLH_2003, IP:ILS_2013} or Brownian forcing of a stochastic ordinary differential equation \cite{MCMC:BPRF_2006}. 
In those settings, efficient MCMC methods have been developed to sample the posterior and compute posterior expectations with convergence rates that are independent of the discretised parameter dimension; \blue{these include} (preconditioned) Crank-Nicolson (pCN) methods \cite{MCMC:BRSV_2008, MCMC:CRSW_2013, MCMC:HSV_2011} that establish the foundation for designing and analysing MCMC algorithms in a function space setting, stochastic Newton methods \cite{MCMC:MWBG_2012,MCMC:Petra_etal_2014} that utilise Hessian information to accelerate the convergence, as well as operator-weighted methods \cite{ MCMC:CLM_2016, MCMC:Law_2014,MCMC:DS_2018} that generalise pCN methods using (potentially location-dependent) operators to adapt to the geometry of the posterior.  

Discretisation also arises in the numerical solution of the forward model, e.g., finite-element discretisation of PDEs.
As many degrees of freedom are needed, it can be computationally demanding to accurately resolve the forward model, which is required to simulate the posterior density. 
A natural way to reduce the computational cost is to utilise a hierarchy of forward models related to a sequence of grid discretisations, ranging from computationally cheaper and less accurate coarse models to more costly but more accurate fine models.
Corresponding to this hierarchy of models, the parameters can also be represented by a sequence of discretised functions with increasing dimensions. 
This yields {\it a hierarchy of posterior distributions}.
By allocating different numbers of MCMC simulations to sample posteriors across different levels and by combining all those sample-based posterior estimations using {\it a telescoping sum} \cite{Sto:Giles_2008}, the multilevel MCMC \cite{MCMC:KST_2013,MCMC:HSS_2013} provides accelerated and unbiased estimates of posterior expectations.

We present a non-trivial integration of the dimension-independent likelihood-informed (DILI) MCMC \cite{MCMC:CLM_2016} and the multilevel MCMC in \cite{MCMC:KST_2013}  to explore the hierarchy of posterior distributions.
This integration offers several advantages:
First, DILI-MCMC employs an intrinsic {\it likelihood-informed subspace} (LIS) \cite{DimRedu:Cui_etal_2014}---which involves a number of forward and adjoint model simulations---to design accelerated operator-weighted proposals. 
By exploiting the multilevel structure of the discretised parameters and discretised forward models, we design a {\it Rayleigh-Ritz procedure} to significantly reduce the computational effort in building a {\it hierarchical LIS} and operating with DILI proposals. 
Second, the resulting DILI-MCMC can drastically improve the sampling efficiency of MCMC at each level, and hence reduce the integration error of multilevel Monte Carlo for a fixed CPU time budget.  
Numerical results confirm the improved computational
efficiency of the proposed multilevel DILI approach.

We note that the DILI proposal has been used before in the  multilevel sequential Monte Carlo (SMC) setting \cite{MCMC:beskos2018multilevel}, but in a very different way. We use derivative information of the likelihood to recursively construct the LIS via matrix--free eigenvalue solves, whereas \cite{MCMC:beskos2018multilevel} uses multilevel SMC to estimate the full-rank empirical posterior covariance matrix and then builds the LIS from this posterior covariance matrix. 
Moreover, we construct DILI proposals by exploiting the structure of the hierarchical LIS to couple Markov chains across levels, whereas \cite{MCMC:beskos2018multilevel} employs the original DILI proposal in the mutation step of SMC to improve mixing.

The paper is structured as follows. Section \ref{sec_backgr} introduces the framework of Bayesian inverse problems and MCMC sampling while section \ref{sec_MLMCMC} discusses the general framework of multilevel MCMC. 
The Rayleigh-Ritz procedure for the recursive construction of the hierarchical LIS is presented in section \ref{sec_MLLIS}. 
The coupled DILI proposals that can exploit the hierarchical LIS are introduced in section \ref{sec_MLDILI}.
Section \ref{sec_num_exp} provides numerical experiments to demonstrate the efficacy of the resulting MLDILI method, while finally, in section \ref{sec_conclusion}, we  provide some concluding remarks.


\section{Background}\label{sec_backgr}

In this section, we review the Bayesian formulation of inverse
problems, the dimension-independent likelihood-informed MCMC approach,
posterior discretisation, as well as the bias-variance
decomposition for MCMC algorithms.   

\subsection{Bayesian inference framework} 
Suppose the parameter of interest is some function $u$ in a separable
Hilbert space $\hilbert(\Omega)$ defined over a given bounded domain
$\Omega \subset \mathbb{R}^d$. 
\blue{We introduce a prior probability measure $\mu_0$ to represent the \textit{a priori} information
about the function $u$.}
The inner product on $\hilbert$ is denoted by $\langle \cdot \, , \cdot \rangle_\hilbert$, with associated norm denoted by $\| \cdot \|_\hilbert$. For brevity, where misinterpretation is not possible, we will drop the subscript $\hilbert$. 
We assume that the prior measure is Gaussian with mean $m_0 \in \hilbert$ and a self-adjoint, positive definite covariance operator $\bGpr$ that is trace-class, so that the prior provides a full probability measure on $\hilbert$.
Given observed data $\data\in\dataspace$ and the forward model $\forward:\hilbert \to \dataspace$, we define the likelihood function $\lkd(\data | u)$ of $\data$ given $u$. %
Denoting the posterior probability measure by $\mu_{y}$, the posterior
distribution  on any infinitesimal volume $du \subseteq \hilbert$ is given by
\begin{equation}\label{eq:post}
	\mu_{y}(du)\propto \lkd(\data|u)\mu_0(du)\,.
\end{equation}
Making the simplifying assumption that the observational noise is
additive and Gaussian with zero mean and covariance matrix $\bGobs$, 
the observation model has the form
\begin{equation}\label{eq:obs}
	\data=\forward(u)+\be,\quad \be\sim\cN(0,\bGobs)\,,
\end{equation}
and it follows immediately that the likelihood function satisfies 
\begin{equation}\label{eq:like}
	\lkd(\data|u) \propto\exp(-\eta(u;\data))\,,
\end{equation}
where $\eta(\data;u)$ is the  data-misfit functional defined by
\begin{equation}\label{eq:misfit}
	\eta(u; \data)\equiv\frac{1}{2} \big(\data-\forward(u)\big)^\top \bGobsi \big(\data-\forward(u)\big)\,.
\end{equation}

\begin{assumption}
\label{assum1}
We assume that the forward model $\forward: \hilbert \rightarrow \dataspace$ satisfies:
\begin{enumerate}
\item For all $\varepsilon > 0$, there exists a constant $K(\varepsilon) > 0$ such that, for all $u \in \hilbert$, 
\[
| \forward(u) | \leq \exp\left( K(\varepsilon) + \varepsilon\|u\|_\hilbert^2 \right) .
\]
\item For any $u \in \hilbert$, there exists a bounded linear operator $\cJ(u): \hilbert \rightarrow \dataspace$ such that
\[
\lim_{\delta u \rightarrow 0} \frac{|\forward(u+\delta u) - \forward(u) - \cJ(u) \delta u|}{\|\delta u\|_\hilbert} = 0, \quad \forall \delta u \in \hilbert.
\]
In particular, this also implies the Lipschitz continuity of $F$.
\end{enumerate}
\end{assumption}
Given observations \blue{$\data \in \mathbb{R}^d$}  
and a
forward model that satisfies Assumption~\ref{assum1}, \cite{IP:Stuart_2010} shows that the resulting data-misfit function is sufficiently bounded and locally Lipschitz, and thus the posterior measure is dominated by the prior measure.
The second condition states that the forward model is first-order Fr\'{e}chet differentiable, and hence the Gauss-Newton approximation of the Hessian of the data-misfit functional is bounded.

Suppose we have some quantity of interest (QoI) that is a functional of the parameter $u$ denoted by $Q:\hilbert\to\R^{q}$, e.g., flow rate. 
Then, posterior-based model predictions can be formulated as expectations of
that QoI over the posterior. We will denote them by 
\[
\Ev_{\mu_y}\big[ Q \big] \equiv \Ev_{U\sim \mu_y} \big[ Q(U) \big].
\]
MCMC methods construct a Markov chain of correlated random variables  $U^{(1)}, \ldots, U^{(N)}$ for which the posterior is the invariant distribution.
Then, one can estimate expected QoI(s) using Monte Carlo integration:
\begin{equation}\label{eq:post_mc}
	\Ev_{\mu_y}\big[ Q \big]  \approx \frac{1}{N} {\textstyle \sum_{i = 1}^{N}} Q(U^{(i)}).
\end{equation}

\subsection{Dimension-independent likelihood-informed MCMC on function space}\label{sec_DILI}

The Metropolis-Hastings (MH) algorithm
\cite{MCMC:Hastings_1970,MCMC:Metropolis_etal_1953} provides a
general framework to design transition kernels that have the posterior as
their invariant distribution to generate a Markov chain of random
variables that targets the posterior.
\begin{definition}
[Metropolis-Hastings Kernel] Given the current state $U^{(k)} =
u^\ast$, a candidate state $u^{\prime}$ can be drawn from a proposal
distribution  $q(u^\ast,\cdot)$.
We define a pair of probability measures 
\begin{equation}
\begin{array}{rll}
\label{eq:nus}
\nu(du^\ast,du^{\prime}) &=& q(u^\ast,du^{\prime}) \mu_y(du^\ast) \\
\nu^\bot(du^\ast,du^{\prime}) &=& q(u^{\prime},du^\ast) \mu_y(du^{\prime}).
\end{array}
\end{equation}
Then, the next state of the Markov chain is set to $U^{(k+1)} =
u^{\prime}$ with probability
\begin{equation} 
\alpha(u^\ast,u^{\prime}) = \min 
\Big \{1, 
\frac{d\nu^\bot}{d\nu}(u^\ast,u^{\prime})
 \Big \} , 
\label{acceptance} 
\end{equation}
and to $U^{(k+1)}=u^\ast$ otherwise.
\end{definition}

MH algorithms require the absolute continuity condition $\nu^\bot \ll
\nu$ to define a valid transition kernel with non-zero acceptance
probability as the dimension goes to infinity
\cite{MCMC:Tierney_1998}.  \blue{We will refer to a MH algorithm as {\it
  well-defined} or {\em dimension-independent} if this absolute continuity condition holds. }
For probability measures over function spaces in the setting
considered here, the sequence of papers
\cite{MCMC:BRSV_2008, MCMC:CRSW_2013,
  MCMC:HSV_2012, MCMC:HSV_2011, IP:Stuart_2010} provide a viable way to construct well-defined MH
algorithms using a preconditioned Crank-Nicolson (pCN) discretisation of a particular
Langevin SDE. 
The pCN proposal has the form 
\begin{equation}
	\label{eq:pcn}
	u^{\prime} = a (u^\ast - m_0) + m_0 -  \gamma (1-a) \prcov \gradientuk + \sqrt{1 - a^2} \prcov^{\frac12} \rand,
\end{equation}
where $\rand \sim \normal(0, \cI) $ and $\gamma \in \{0, 1\}$ is a
tuning parameter to switch between Langevin ($\gamma = 1$) and Ornstein-Uhlenbeck proposal ($\gamma = 0$). It is required that $a \in (-1,1)$. 
The pCN proposal \eqref{eq:pcn} satisfies the desired absolute
continuity condition and the acceptance probability does not go to
zero as the discretisation of $u$ is refined.
\blue{In addition, \cite{hairer2014spectral,MCMC:DS_2018} establish that, under mild conditions, the spectral gaps of the MCMC transition kernels defined by (generalised) pCN proposals do exist, and that they are independent of the dimension of the discretised parameters. Thus, the statistical efficiency of pCN proposals is also dimension-independent.}

The pCN proposal \eqref{eq:pcn} scales uniformly in all directions with respect to the norm induced by the prior covariance. 
Since the posterior necessarily contracts the prior
along parameter directions that are informed by the likelihood, the
Markov chain produced by the standard pCN proposal decorrelates more
quickly in the likelihood-informed parameter subspace than in the
orthogonal complement, which is
prior-dominated \cite{MCMC:CLM_2016,MCMC:Law_2014}. 
Thus, proposed moves of pCN can be effectively too small 
in prior-dominated directions, resulting in poor mixing.

\newcommand{\dilipotential}[1]{ {- \potential(#1; y) - \frac12 \langle #1, \cB^{-2} (\cA^2 + \cB^2 - \cI^{}) #1 \rangle} }

The dimension-independent likelihood-informed (DILI) MCMC \cite{MCMC:CLM_2016} provides a systematic way to design proposals that adapt to the anisotropic structure of the posterior while retaining dimension-independent performance.
It considers operator-weighted proposals in the form of
\begin{equation}
u' =  \prcov^{\frac12}\cA\prcov^{-\frac12} (u^\ast-m_0) + m_0 - \prcov^{\frac12} \cG \prcov^{\frac12} \gradientuk + \prcov^{\frac12}\cB \rand,
\label{eq:operator_weighted}
\end{equation}
where $\cA$, $\cB$, and $\cG$ are bounded, self-adjoint operators on $\im(\prcov^{-1/2})$ that satisfy certain properties to be discussed below.
In this paper, we set $\cG$ to zero throughout and thus consider only non-Langevin type proposals. By applying a whitening transform 
\begin{equation}
v = \prcov^{-\frac12} (u - m_0)
\label{eq:trans_v}
\end{equation} 
to the parameter $u$ and by denoting (in
a slight abuse of notation) the associated data-misfit functional again by
$\potential(v;y) \leftarrow \potential(\prcov^{1/2} v + m_0 ; y)$, the
proposal \eqref{eq:operator_weighted} simplifies to
\begin{equation}
v' =  \cA v^\ast + \cB \rand .
\label{eq:operator_weighted_v}
\end{equation}
The following theorem provides sufficient conditions for constructing
the operators $\cA$ and $\cB$ such that the proposal
\eqref{eq:operator_weighted_v} yields a well-defined MH algorithm, as
well as a formula for the acceptance probability.
\begin{theorem} 
\label{theo:1}
Suppose that the posterior measure $\mu_y$ is equivalent to the prior
measure $\mu_0$ and that the self-adjoint operators $\cA$ and $\cB$
commute, that is, they can be defined by a common set of
eigenfunctions $\{\psi_i \in
\im(\prcov^{-1/2}) : i \in \mathbb{N}\}$ with corresponding
eigenvalues $\{a_i\}_{i = 1}^{\infty}$ and $\{b_i\}_{i = 1}^{\infty}$,
respectively. Suppose further that
\[
\{a_i\}_{i = 1}^{\infty},\ \{b_i\}_{i = 1}^{\infty} \subset \R \backslash \{0\}
\quad \text{and} \quad \sum_{i = 1}^\infty \left( a_i^2 + b_i^2 - 1
\right)^2  < \infty.
\]
Then, the proposal \eqref{eq:operator_weighted_v} delivers a well-defined
MCMC algorithm  and the acceptance probability is given by
\[
\alpha\big(v^\ast, v'\big) = \min \bigg\{ 1, \frac{\exp\big(\dilipotential{v^{\prime\,}}\big)}{\exp\big(\dilipotential{v^\ast}\big)} \bigg\}.
\]
\end{theorem}
\begin{proof}
The above assumptions are simplified versions of those in
Theorem~3.1 of~\cite{MCMC:CLM_2016}. The acceptance probability
directly follows from Corollary~3.5 of~\cite{MCMC:CLM_2016}.
\end{proof}

The DILI proposal \eqref{eq:operator_weighted_v} enables different scalings in the proposal moves along different parameter directions. By choosing appropriate eigenfunctions $\{\psi_i\}_{i =1}^{\infty}$ and eigenvalues $\{a_i, b_i\}_{i = 1}^{\infty}$, it can capture the geometry of the posterior, and thus can potentially improve the mixing of the resulting Markov chain. 

The likelihood-informed subspace (LIS)
\cite{DimRedu:Cui_etal_2014,ROM:CMW_2016} provides a viable way to
construct such operators $\cA$ and $\cB$.  
It is spanned by the leading eigenfunctions of the eigenvalue problem
\begin{equation}
\expect_{V \sim \mu^\ast} \big[ \hessian(V) \big]  \psi_i = \lambda_i \psi_i,
\label{eq:eig}
\end{equation}
where $\hessian(v)$ is some information metric of the likelihood
function (with respect to the transformed parameter $v$), for example,
the Hessian of the data-misfit functional
or the Fisher information, and $\mu^\ast$ is some reference measure,
for example, the posterior or the Laplace approximation of the
posterior.
In the LIS, spanned by $\{\psi_1, \ldots, \psi_r\}$, the posterior may
significantly differ from the prior. Thus, we prescribe inhomogeneous
eigenvalues $\{a_i\}_{i =1}^r$ and $\{b_i\}_{i = 1}^r$ to ensure that
the proposal follows the possibly relatively tight geometry of the posterior.
In the complement of the LIS, where the posterior does not differ
significantly from the prior, we can use the original pCN proposal and
set $\{a_i\}_{i>r}$  and $\{b_i\}_{i > r}$ to some constant values
$a_\perp$ and $b_\perp$, respectively.
Further details on the computation of the LIS basis and the choice of
eigenvalues will be discussed in the multilevel context in later
sections.

\subsection{Posterior discretisation and bias-variance decomposition} 
\label{sec:discretised_post}
When the forward model involves a partial/ordinary differential
equation and the parameter is defined as a
spatial/temporal stochastic process, 
it is necessary in practice to discretise 
the parameter and the forward
model using appropriate numerical methods.

A common way to discretise the parameter is the Karhunen--Lo\'{e}ve
expansion, which also serves the purpose of the whitening transform. 
Given the prior mean $m_0(x)$ and the prior covariance $\prcov$, we
express the unknown parameter $u$ as the linear combination of the first $R$
eigenfunctions $ \{\phi_1, \ldots, \phi_{R}\}$ of the eigenvalue
problem $\prcov \phi_j = \omega_j \phi_j$, such that
\begin{equation}\label{eq:KL}
\textstyle u_R(x) = m_0(x) + \sum_{j = 1}^R \sqrt{\omega_j} \,\phi_j(x)\, v_j, \quad x \in \Omega.
\end{equation}
The discretised prior $p_R(\bv)$ associated with the random
coefficients $\bv = [v_1, \ldots, v_R]^\top$ is Gaussian with zero
mean and covariance equal to the $R\times R$ identity matrix $\cI_R$. \blue{In this context, the selection of the truncation dimension $R$ is typically based on the rate of decay of the eigenvalues $\omega_j$, for example such that $\sum_{j=1}^R \omega_j / \sum_{j=1}^\infty \omega_j \geq \tau$ with a threshold $\tau \in (0,1)$ close to one. In this way, the truncated representation encapsulates a specific percentage of the total prior variance.}

We discretise the forward model using a numerical method, such as finite elements or finite differences, with $M$ degrees of freedom, which yields a discretised forward model $F_{R,M}$ mapping from the discretised coefficients $\bv$ to the observables. 
In this way, the posterior measure \eqref{eq:post} can be discretised,
leading to the finite-dimensional density 
\begin{equation}
\pi_{R,M}(\bv | \data) \propto \exp( -\potential_{R,M}(\bv; \data) )\, p_R(\bv),
\end{equation}
where 
\[
\potential_{R,M}(\bv; \data) = \frac{1}{2} \big(\data-\forward_{R,M}(\bv)\big)^\top \bGobsi \big(\data-\forward_{R,M}(\bv)\big)
\]
is the discretised data-misfit function. 
Correspondingly, we also define the discretised QoI $Q_{R,M}(\bv)$, which
maps the discretise coefficient vector $\bv$ to the discretised QoI. 

The discretised parameters and forward models can be indexed by the
discretisation level. We consider a hierarchy of $L+1$ levels of
discretised parameter spaces with dimensions $R_0 \!\leq\! R_1
\!\leq\! \ldots \!\leq\! R_L$ and a hierarchy of discretised forward
models with $M_0 \!\leq\! M_1 \!\leq\! \ldots\!\leq\! M_L$ degrees of freedom.
Discretised parameter, forward model and QoI on level $\ell$ are
denoted by 
$$
\bv_\ell = [v_1, \ldots, v_{R_\ell}]^\top, \quad F_\ell(\bv_\ell)
\equiv  F_{R_\ell,M_\ell}(\bv_\ell) \ \ \text{and} \ \
Q_\ell(\bv_\ell) \equiv  Q_{R_\ell,M_\ell}(\bv_\ell), 
$$
respectively. 
Thus, the discretised data-misfit function, prior and posterior on
level $\ell$ are
\begin{equation}
\potential_\ell(\bv_\ell; \data) \!\equiv\! \potential_{R_\ell,M_\ell}(\bv_\ell; \data), \;\; 
p_\ell(\bv_\ell) \!\equiv\! p_{R_\ell}(\bv_\ell) , \; \textrm{and} \;\;
\pi_\ell(\bv_\ell | \data) \!\equiv\! \pi_{R_\ell,M_\ell}(\bv_\ell | \data),
\end{equation}
respectively, with the associated posterior expectation $\Ev_{\pi_\ell}\big[Q_\ell\big] \equiv \Ev_{ \rvl \sim \pi_\ell} \big[ Q_\ell^{}(\rvl^{}) \big]$.

\begin{assumption}\label{assum_bias}
\begin{enumerate}
\item[(i)] The bias of the posterior expectation on level
$\ell$ can be bounded in terms of the number of degrees of freedom of the forward model such that
\begin{equation}\label{eq:bias_hp}
\big|\Ev_{\mu_y}\big[Q\big] - \Ev_{\pi_\ell}\big[Q_\ell\big] \big|  = \cO(M_\ell^{-\vartheta_{\rm b}}), 
\end{equation}
for some constant $\vartheta_{\rm b} >0$. 
\item[(ii)] For the computational cost of carrying out one step of MCMC (including a
forward model simulation) it is assumed that there exists a constant
$\vartheta_{\rm c} > 0$ such that
\begin{equation}\label{eq:cost_hp}
C_\ell = \cO(M_\ell^{\vartheta_{\rm c}}).
\end{equation}
\end{enumerate}
\end{assumption}

Implicitly, Condition (\romannumeral 1) in Assumption \ref{assum_bias} also assumes that $R_\ell$ is sufficiently large such that on level $\ell$ the bias due to parameter approximation is dominated by the error due to the forward model approximation.
This condition can be verified for certain classes of model problems. For instance, for finite element methods applied to elliptic PDEs (which is the model problem used in the numerical experiments of this work), the convergence analysis in \cite[Section 4.2]{MCMC:KST_2013} shows that the discretisation error satisfies that $|\Ev_{\mu_y}[Q] - \Ev_{\pi_{\ell}}[Q_{\ell}]| = \mathcal{O}(M_\ell^{-\vartheta_{\rm b}} + R_{\ell}^{-\vartheta_{\rm b}'})$ for some constants $\vartheta_{\rm b}, \vartheta_{\rm b}'>0$. Thus, by choosing $R_\ell = M_\ell^{\vartheta_{\rm b}/\vartheta_{\rm b}'}$, the two error contributions are balanced. The constant $\vartheta_{\rm c}$ in Condition (\romannumeral 2) of Assumption \ref{assum_bias} depends on the underlying linear solver and/or numerical integrator, so that a theoretical upper bound on $\vartheta_{\rm c}$ is often known. 

Consider discretisation level $L$ and let
$\{\bV_{\!\!L}^{(j)}\}_{j= 1}^{N_{\textnormal{MC}}}$ be a Markov
chain produced by a MCMC algorithm converging in distribution to
$\pi_L$. An estimate for the expectation $\Ev_{\pi_L}\big[Q_L\big]$ is
\begin{equation}\label{eq:YMC}
Y^{\textnormal{MC}} \equiv \frac{1}{N_{\textnormal{MC}}}
{\textstyle \sum_{j=1}^{N_{\textnormal{MC}}} } Q_L^{}
(\bV_{\!\!L}^{(j)}) \ \approx \ \Ev_{\pi_L}\big[Q_L\big]\,.
\end{equation}
The focus of this work is the asymptotic
performance of algorithms, and hence the initialization bias of
  MCMC and the computational cost due to burn-in are not
discussed. The mean-squared-error (MSE) of the Monte Carlo estimator
\eqref{eq:YMC} allows a bias-variance decomposition of the form
\begin{equation}
\label{eq:bias-var-decomp}
\textnormal{MSE}(Y^{\textnormal{MC}}) = \underbrace{\Big| \Ev_{\mu_y}\big[Q\big] - \Ev_{\pi_L}\big[Q_L\big] \Big|^2}_{\rm Square\;of\;Bias} + \underbrace{\Var_{\pi_L}(Q_L) \Big/ N^{\textnormal{eff}}_{\textnormal{MC}}}_{\Var(Y^{\textnormal{MC}})},
\end{equation}
where $N^{\textnormal{eff}}_{\textnormal{MC}}$ is the effective sample
size of the Markov chain $\{\bV_{\!\!L}^{(j)}\}_{j= 1}^{N_{\textnormal{MC}}}$. 
This effective sample size is proportional to the total sample size,
i.e., $N_{\textnormal{MC}}^{\textnormal{eff}} = N_{\textnormal{MC}}/
\tau_{\textnormal{MC}}$, where $\tau_{\textnormal{MC}} \geq 1$ is the
integrated autocorrelation time (IACT) of the Markov
chain. The work of \cite{hairer2014spectral} shows that for pCN-type algorithms the IACT $\tau_{\textnormal{MC}}$ is dimension-independent given a local Lipschitz assumption, which is often satisfied for inverse problems governed by elliptic PDEs. 

Choosing $N_{\textnormal{MC}}$ such that the two terms in 
\eqref{eq:bias-var-decomp} of the MCMC estimator are balanced and
using Assumption \ref{assum_bias}, the
total  computational cost to  achieve $\textnormal{MSE}(Y^{\textnormal{MC}}) <
\varepsilon^2$ is
\begin{equation}\label{eq:MCMCcompcost}
C^{\textnormal{MC}} = \cO(\tau_{\textnormal{MC}}\,\varepsilon^{-2-\vartheta_{\rm c}/\vartheta_{\rm b} })\,.
\end{equation}
Thus, one of the key aims in accelerating MCMC sampling is to reduce $\tau_{\textnormal{MC}}$,
which can be achieved, e.g., via DILI MCMC proposals. In addition, the multilevel
method will allow us to improve the asymptotic rate of growth of the
cost of the standard MCMC estimator in \eqref{eq:MCMCcompcost} with
respect to $\varepsilon$, as well as to further reduce $\tau_{\textnormal{MC}}$ on the
higher levels. These two things are achieved in multilevel
MCMC, by using coarse level samples as proposals on the higher levels and by dealing
with the high numerical correlation between subsequent MCMC samples
produced by standard proposal mechanisms on the coarsest level (level
zero).  Thus, most samples are drawn on the computationally least
costly level zero, as well as shifting most of the work for removing the
initialization bias to level zero, all contributing to the practical
advatages of the multilevel algorithms compared to their single-level counterparts.


\section{Multilevel MCMC}\label{sec_MLMCMC}
By exploiting the hierarchy of posteriors, the rate of the
computational cost in \eqref{eq:MCMCcompcost} can be reduced significantly
using the multilevel idea in \cite{MCMC:KST_2013}. We
expand the  posterior
expectation in the telescoping sum
\begin{equation}\label{tel_sum}
\textstyle \Ev_{\pi_L}[Q_L^{}] = \Ev_{\pi_0}[Q_0^{}] + \sum_{\ell=1}^L\Big(\Ev_{\pi_\ell}[Q_\ell^{}]-\Ev_{\pi_{\mell}}[Q_{\mell}^{}]\Big)\,.
\end{equation}

For level zero, the sample set $\{\! \bV_{\!\!0}^{(0,\,j)}
\!\}_{j=1}^{N_0}$ is assumed to be drawn via some MCMC method that
converges to $\pi_0(\,\cdot\,|\data)$ and the first term in the
telescoping sum \eqref{tel_sum} is estimated via
\begin{align*}
Y_0 \equiv \frac{1}{N_0}{\textstyle\sum_{j=1}^{N_0}} D_0^{(j)} \
  \approx \  \Ev_{\pi_0}[Q_0], \quad \text{where} \quad D_0^{(j)} = Q^{}_0(\bV_{\!\!0}^{(0,\,j)}).
\end{align*}

Since the two expectations in the difference
$\Ev_{\pi_\ell}[Q_\ell^{}]-\Ev_{\pi_{\mell}}[Q_{\mell}^{}]$ are with
respect to different discretisations of the posterior, special
treatment is required for $\ell > 0$.
Let $\Delta_{\ell, \mell}(\bv_{\ell},
\bv_{\mell})$ be the joint density of $\bv_{\ell}$ and $\bv_{\mell}$ such that 
\begin{equation}\label{eq:joint_coupling}
\quad \int \!\!\Delta_{\ell, \mell}(\bv_{\ell}, \bv_{\mell})
\,d\bv_{\mell} \!=\! \pi_\ell(\bv_{\ell}| \data) \ \ \textrm{and\ } \ 
\int\!\! \Delta_{\ell, \mell}(\bv_{\ell}, \bv_{\mell}) \, d \bv_{\ell}
\!=\! \pi_{\mell}(\bv_{\mell}| \data), 
\end{equation}
that is, the posteriors $\pi_\ell(\bv_{\ell}| \data)$ and
$\pi_{\mell}(\bv_{\mell}| \data)$ are the two marginals. Then, the
difference between expectations can be expressed as
\begin{equation}
\Ev_{\pi_\ell}[Q^{}_\ell]-\Ev_{\pi_{\mell}}[Q^{}_{\mell}] =
\Ev_{\Delta_{\ell, \mell}} [ D_\ell ] , \quad \textrm{where} \quad  D_\ell = Q_\ell(\rvl) - Q_{\mell}(\rvml)
\end{equation}
and $(\rvl, \rvml ) \sim \Delta_{\ell, \mell}(\cdot, \cdot)$.
The construction of the joint density and the associated sampling
procedure will be critical to reduce the computational complexity.

Suppose the samples $\big\{ \big(\rvl^{(\ell, j)}, \rvml^{(\ell,
  j)}\big)\big\}_{j=1}^{N_\ell}$ form a Markov chain that converges in
distribution to $\Delta_{\ell, \mell}(\cdot, \cdot)$ and 
\[
D_\ell^{(j)} = Q^{}_\ell\big(\rvl^{(\ell,j)}\big)-Q_{\mell}^{}\big(\rvml^{(\ell,j)}\big).
\]
Then, the remaining terms in \eqref{tel_sum}, for $\ell =1, \ldots,L$, are estimated by 
\begin{align*}
Y_\ell \equiv \frac{1}{N_\ell}{\textstyle\sum_{j=1}^{N_\ell}}
  D_\ell^{(j)}  \ \approx \ \Ev_{\pi_\ell}[Q^{}_\ell]-\Ev_{\pi_{\mell}}[Q^{}_{\mell}]
\end{align*}
and the multilevel MCMC estimator for $\Ev_{\pi_L}[Q_L]$ is defined by
\begin{align}
\quad\;\; \textstyle \Ev_{\pi_L} \big[ Q_L^{} \big] \approx Y^{\textnormal{ML}} &\equiv {\textstyle\sum_{\ell=0}^L} Y_\ell\,,
\end{align}
The mean square error of this estimator can again be decomposed as follows:
\begin{align}
\textnormal{MSE}(Y^{\textnormal{ML}}) & \equiv  \underbrace{\big|\Ev_{\mu_y}\big[Q\big] - \Ev_{\pi_{L}}[Q_L]\big|^2}_{\rm Square\;of\;Bias} + \underbrace{\textstyle  \sum_{\ell=0}^L \big( \Var(Y_\ell) + \, \sum_{k \neq \ell}^L \Cov( Y_\ell, Y_k ) \big)}_{\Var(Y^{\textnormal{ML}})}. \label{MSE_ML}
\end{align}

\subsection{Variance management}

For optimal efficiency, we now choose the numbers of samples $N_\ell$,
$\ell=0,\ldots,N$, such as to minimise $\Var(Y^{\textnormal{ML}})$ for fixed computational effort. This includes the {\it within-level variance} $\Var(Y_\ell)$ and the {\it cross-level variance} $\Cov( Y_\ell, Y_k )$ for $k \neq \ell$. 
We will provide justifications on managing these variances using the following assumptions.

\begin{remark}\label{remark:var_cov}
Suppose the effective sample sizes are proportional to the total
sample sizes, i.e.,
$N_\ell^{\textnormal{eff}} = N_\ell \big/ \tau_{\ell}$, for all
$\ell$, where $\tau_{\ell} \geq 1$ is the IACT of the Markov chain $D_\ell^{(j)}$. Then, the within-level variance has the form
\begin{equation}\label{var_ML}
\Var(Y_\ell) = \frac{1}{N_\ell^{\textnormal{eff}}} \Var_{\Delta_{\ell,\mell}}(D_\ell) =  \frac{\tau_{\ell}}{N_\ell} \Var_{\Delta_{\ell,\mell}}(D_\ell) , 
\end{equation}
where we set $\Var_{\Delta_{0,-1}}(D_0) = \Var_{\pi_0}(Q_0)$ and have
\begin{align*}
\Var_{\Delta_{\ell,\mell}}(D_\ell) = \Var_{\pi_{\ell}}(Q_\ell) + \Var_{\pi_{\mell}}(Q_{\mell}) - 2\textrm{Cov}_{\Delta_{\ell,\mell}}(Q_\ell, Q_{\mell}) \geq 0, \;\; \forall \ell > 0,
\end{align*}
by the Cauchy--Schwarz inequality. 
Thus, to reduce $\Var(Y_\ell)$, the joint density should be constructed
in such a way that $\textrm{Cov}_{\Delta_{\ell,\mell}}(Q_\ell,
Q_{\mell})$ is {\it positive} and (if possible) maximised. In
addition, the MCMC
simulation should be made {\it statistically efficient} in the sense that $\tau_{\ell}$ is as close to one as possible.
\end{remark}

\begin{assumption}\label{assum_within}
The variance $\Var_{\Delta_{\ell,\mell}}(D_\ell)$ converges to zero as
$M_\ell \rightarrow \infty$ and 
\begin{equation}
\Var_{\Delta_{\ell,\mell}}(D_\ell) = \cO(M_\ell^{-\vartheta_{\rm v} })\,, 
\end{equation}
for some constant $\vartheta_{\rm v} > 0$.
\end{assumption}

\begin{proposition}\label{assum_cross}
Suppose that there exists an $r < 1$ such that 
\begin{equation}\label{eq:var_cov_assum}
\frac{\Cov( Y_\ell, Y_k )}{\max\{\Var(Y_\ell), \Var(Y_k)  \}} \leq
r^{|k-l|}, \quad \text{for all} \ \ k
\neq \ell,
\end{equation}
i.e., the cross-level covariance is insignificant compared to the
within-level variance. Then
\begin{equation}\label{eq:var_bound}
\Var(Y^{\textnormal{ML}}) = \sum_{\ell=0}^L  \Big(
\Var(Y_\ell) + \, \sum_{k \neq \ell}^L \Cov( Y_\ell, Y_k ) \Big) \le 
\frac{1+r}{1-r} \sum_{\ell=0}^L \Var(Y_\ell) \,.
\end{equation}
\end{proposition}
\begin{proof}
Without loss of generality, we can assume the variances $\{ \Var(Y_\ell) \}_{\ell=0}^L$ are ordered as $\Var(Y_\ell) \geq \Var(Y_k)$ for $\ell < k$. Then we have the bound
\begin{align*}
\textstyle \Var(Y^{\textnormal{ML}}) & = \sum_{\ell=0}^L  \Big( \Var(Y_\ell) + 2 \,\sum_{k > \ell}^L \Cov( Y_\ell, Y_k ) \Big) \\
& \le \sum_{\ell=0}^L  \Var(Y_\ell) \Big( 1 + 2 \, \sum_{k > \ell}^\infty \frac{\Cov( Y_\ell, Y_k )}{\Var(Y_\ell)} \Big) \\
& \le \sum_{\ell=0}^L  \Var(Y_\ell) \Big( 1 + 2 \, \sum_{k >
  \ell}^\infty r ^{(k-\ell)} \Big) \ = \ \frac{1+r}{1-r} \sum_{\ell=0}^L  \Var(Y_\ell).
\end{align*}

\end{proof}

Using Proposition \ref{assum_cross} and \eqref{var_ML}, the variance of the multilevel estimator satisfies
\[
\Var(Y^{\textnormal{ML}}) = \cO \Big( {\textstyle \sum_{\ell=0}^L}  \frac{\tau_{\ell} }{ N_\ell } \Var_{\Delta_{\ell,\mell}}(D_\ell) \Big).
\]
The total computational cost is $C^{\textnormal{ML}} = \sum_{\ell=0}^L N_\ell\,C_\ell$.
This way, for a fixed variance, the computational cost is minimised by choosing the sample size
\begin{align}\label{n_l}
\textstyle N_\ell \propto \sqrt{ \tau_{\ell}\,\Var_{\Delta_{\ell,\mell}}(D_\ell) \big/ C_\ell },
\end{align}
which leads to a total computational cost that satisfies
\begin{align}\label{cost_L}
C^{\textnormal{ML}} \propto {\sum_{\ell=0}^L} \sqrt{ \tau_{\ell}\,C_\ell\,\Var_{\Delta_{\ell,\mell}}(D_\ell) }.
\end{align}

\begin{theorem}\label{thm:mlmcmc}
Suppose Assumptions \ref{assum_bias}, \ref{assum_within} and \eqref{eq:var_cov_assum} in Proposition \ref{assum_cross} hold. For the multilevel MCMC estimator to satisfy
$\textnormal{MSE}(Y^{\textnormal{ML}}) < \varepsilon^2$, the
multilevel MCMC with $N_\ell$ chosen as in \eqref{n_l} requires an
overall computational cost
\begin{equation}\label{cost_ML}
C^{\textnormal{ML}} = \,\begin{cases}
\cO(\varepsilon^{-2}) &\mbox{if } \vartheta_{\rm v}>\vartheta_{\rm c}\\
\cO(\varepsilon^{-2}|\log\varepsilon|^2) &\mbox{if } \vartheta_{\rm v}=\vartheta_{\rm c} \\
\cO(\varepsilon^{-2-(\vartheta_{\rm c} - \vartheta_{\rm v})/\vartheta_{\rm b} }) &\mbox{if } \vartheta_{\rm v}<\vartheta_{\rm c}
\end{cases}\,.
\end{equation}
\end{theorem}

\begin{proof}
Follows directly from the multilevel Monte
Carlo complexity theorems in \cite {cgst11,Sto:Giles_2008}.
\end{proof}

It is difficult to rigorously verify Assumption
\eqref{eq:var_cov_assum} in Proposition \ref{assum_cross}, but it
is often observed that the cross-level variances $\Cov( Y_\ell, Y_k )$
rapidly decay to zero in practice, as the Markov chains used for
computing $Y_\ell$ and $Y_k$ with $\ell \neq k$ are statistically
independent. For example, in
\cite{MIMCMC:JKLZ_2018} independent Markov chains are constructed and in \cite{MCMC:KST_2013} a subsampling strategy of the coarser chains is employed to ensure independence.
Nevertheless, the bound on
the computational complexity of multilevel MCMC is reduced under
assumption \eqref{eq:var_cov_assum} compared to that presented in~\cite{MCMC:KST_2013}, which 
has an extra $|\log\varepsilon|$ factor. 
For any positive values of $\vartheta_{\rm b},\vartheta_{\rm v},\vartheta_{\rm c}$, the multilevel MCMC approach asymptotically requires less computational effort than
single-level MCMC. 
To choose optimal numbers of samples on the various levels, estimates of the
IACTs $\tau_\ell$, the variances
$\Var_{\Delta_{\ell,\mell}}(D_\ell)$, and the computational costs
$C_\ell$ are needed. Such quantities may not be known {\it a priori},
but they can all be obtained and adaptively improved (on the fly) as
the simulation progresses.

\subsection{Notations}\label{sec:notation}
To map vectors and matrices across adjacent levels of discretisation
we define the following notation. 
Given the canonical basis $(\hat \be_1, \hat\be_2, \dots,
\hat\be_{R_\ell})$ of the parameter space at level $\ell$, where
$\hat\be_j \in \R^{R_\ell}$, we define the basis matrices
$\Theta_{\ell, c} \equiv (\hat\be_1, \hat\be_2,\dots,
\hat\be_{R_{\mell}})$ and $\Theta_{\ell, f} \equiv (\hat
\be_{R_{\mell}+1}, \dots, \hat \be_{R_\ell})$, which correspond to the
parameter coefficients `active' at level $\mell$ and the additional coefficients. \blue{Here the subscripts $c$ and $f$ denote the coefficients that are 'active' on the coarse level and the coefficients that are 'active' only on the fine level, respectively. }
We can split the parameter $\bv_\ell$ into two components 
\begin{equation}
\bv_\ell = \begin{bmatrix} \bv_{\ell, c} \\ \bv_{\ell, f}\end{bmatrix}, \;\; \textrm{where}\;\; \bv_{\ell, c} =  \Theta_{\ell, c}^\top \bv_\ell \;\;\textrm{and}\;\; \bv_{\ell, f} =  \Theta_{\ell, f}^\top \bv_\ell,
\end{equation}
which correspond to the coefficients on the previous level $\mell$ and the additional coefficients.
Given a matrix $\cA_\ell \in \R^{R_\ell \times R_\ell}$, we partition the matrix as
\begin{equation}
\cA_\ell=\begin{bmatrix}
\cA_{\ell,cc} &\cA_{\ell,cf}\\
\cA_{\ell,fc} &\cA_{\ell,\ff}
\end{bmatrix},
\end{equation}
where $\cA_{\ell,cc} \equiv  \Theta_{\ell, c}^\top \cA_\ell
\Theta_{\ell, c}$ and $\cA_{\ell,\ff}$, $\cA_{\ell,fc}$ and
$\cA_{\ell,cf}$ are defined analogously.
The matrices $\Theta_{\ell, c}$ and $\Theta_{\ell, f}$ are never 
constructed explicitly. Operations with those matrices only involve the selection
of the corresponding rows or columns of the matrix or vector.


\newcommand{\lis}{{\boldsymbol \psi}}
\renewcommand{\bphi}{{\boldsymbol \phi}}

\section{Multilevel LIS}\label{sec_MLLIS}

\blue{We aim to employ the DILI method (cf. Section~\ref{sec_DILI}) as the proposal mechanism for multilevel MCMC. Since the computation of the LIS basis used by the DILI proposal can be costly, here we develop a Rayleigh--Ritz procedure to recursively compute new, multilevel likelihood-informed subspaces using the model hierarchy. The resulting hierarchical LIS basis can be used to generalise DILI proposals to the multilevel setting and to improve the efficiency of multilevel MCMC sampling. In Section \ref{sec:lis_setup}, we define the concept of LIS in the multilevel context. In Sections  \ref{sec:lis_base} and \ref{sec:lis_enrich}, we present the recursive construction of the 
multilevel LIS using the Rayleigh--Ritz procedure.}

\subsection{Setup}\label{sec:lis_setup}
For each level $\ell \in \{0, 1, \ldots, L\}$, we denote the linearisation of the forward model $\forward_\ell$ at a given parameter $\bv_\ell$ by
\[
\linear_\ell(\bv_\ell) = \nabla_{\bv_\ell} \forward_\ell(\bv_\ell).
\] 
This yields the Gauss-Newton approximation of the Hessian of the data-misfit functional at $\bv_\ell$ (hereafter referred to as the Gauss--Newton Hessian) in the form of
\begin{equation}
\hessian_\ell(\bv_\ell) = \linear_\ell(\bv_\ell)^\top \obscov^{-1} \,  \linear_\ell(\bv_\ell).
\label{eq:gn_hessian}
\end{equation}
The Gauss--Newton Hessian in \eqref{eq:gn_hessian} corresponds to the Fisher information matrix of the likelihood with additive Gaussian noise. It is commonly used in statistics to measure the local sensitivity of the parameter-to-likelihood map. The leading eigenvectors of $\hessian_\ell(\bv_\ell)$ (corresponding to the largest eigenvalues) indicate parameter directions along which the likelihood function varies rapidly. 

However, to extract the global sensitivity of the
  parameter-to-likelihood map from the local sensitivity information
  contained in the Gauss--Newton Hessian, it is necessary to compute the expectation of
  $\hessian_\ell(\bv_\ell)$ with respect to some reference
  distribution $p_\ell^\ast(\bv_\ell)$, i.e.,
\begin{equation}
\expect_{\rvl^{} \sim p_\ell^\ast} \big[ \hessian_\ell(\rvl) \big].
\end{equation}
Finally,  this is approximated using the sample average with $K_\ell$ random samples drawn from the reference distribution, which yields
\begin{equation}
\expect_{\rvl^{} \sim p_\ell^\ast} \big[ \hessian_\ell(\rvl) \big] \approx \widehat\hessian_\ell \equiv \frac{1}{K_\ell} {\sum_{k = 1}^{K_\ell} } \hessian_\ell(\bv_\ell^{(k)}), \;\, \textrm{where} \;\; \bv_\ell^{(k)} \sim p_\ell^\ast(\cdot).
\label{eq:exp_Hl}
\end{equation}
Note that the matrix $\widehat\hessian_\ell$ is symmetric and positive semidefinite.
Different choices of the reference distribution, such as the prior or the posterior, lead to different ways to construct the LIS and different performance characteristics.

\begin{remark}
Following the discussion in \cite{ROM:CMW_2016,DimRedu:Zahm_etal_2018}, using the posterior as the reference leads to sharp approximation properties \cite{cui2022unified,DimRedu:Zahm_etal_2018} compared to other choices. However, the posterior exploration relies on MCMC sampling, and thus this choice requires adaptively estimating LIS during the MCMC sampling. 
The Laplace approximation to the posterior provides a reasonable
alternative in a wide range of problems where the posterior is
unimodal. We use the Laplace approximation as the reference
distribution in this work.  

The choice of the reference distribution can have an impact
on the quality of the LIS basis and on the IACT of the Markov chains
produced by DILI MCMC, but it does not affect the convergence of MCMC, 
as DILI samples the full parameter space and only uses the LIS to
reduce the IACT and thus to accelerate posterior sampling.
\end{remark}

It is often computationally infeasible to explicitly form the Gauss--Newton Hessian matrix \eqref{eq:gn_hessian}. 
However, all we need are matrix-vector-products with the Gauss--Newton Hessian
matrix. This requires only applications of the linearised forward
model $\linear_\ell(\bv_\ell)$ and its adjoint
$\linear_\ell(\bv_\ell)^\top$, which are well-established operations
in the PDE-constraint optimisation literature. We refer the readers to
recent applications in Bayesian inverse problems for further details, e.g., \cite{IP:BGMS_2013,MCMC:MWBG_2012,MCMC:Petra_etal_2014}.

\subsection{Base level LIS} \label{sec:lis_base}
At the base level, we use the samples $\{\bv_0^{(k)}\}_{k=1}^{K_0}$ drawn from the reference $p_0^\ast(\cdot)$ to construct the sample-averaged Gauss--Newton Hessian, $\widehat\hessian_0$.
Then, we use the Rayleigh quotient $\langle \boldsymbol \phi, \widehat\hessian_\ell \, \boldsymbol \phi \rangle \,/\, \langle \boldsymbol \phi, \boldsymbol \phi \rangle$ to measure the (quadratic) change in the parameter-to-likelihood map along a parameter direction $\boldsymbol \phi$.
Hence, the LIS can be identified via a sequence of optimisation problems of the form 
\begin{equation}\label{eq:Rayleigh_0}
\lis_{0,k+1} = \argmax_{\| \boldsymbol \phi \| = 1} \langle
\boldsymbol \phi, \widehat\hessian_0 \, \boldsymbol \phi \rangle, \
\textrm{\ subject\ to\ }  \ \langle \boldsymbol \phi, \lis_{0,i}
\rangle = 0, \quad \textrm{\;for\;} i = 1,  \ldots, k,
\end{equation}
where $\lis_{0,1}$ is the solution to the unconstrained optimisation problem.
The sequence of optimisation problems in \eqref{eq:Rayleigh_0} is
equivalent to finding the leading eigenvectors of 
$\widehat\hessian_0$.

\begin{definition}[Base level LIS]
Given the sample--averaged Gauss--Newton Hessian $\widehat\hessian_0$ on level $0$ and a
threshold $\eigentrunc > 0$, we solve the eigenproblem
\begin{equation}
\widehat\hessian_0 \, \lis_{0,i} = \lambda_{0,i} \lis_{0,i},
\label{eq:exp_H0}
\end{equation}
and then use the $r_0$ leading eigenvectors with eigenvalues $\lambda_{0,i} > \eigentrunc$, for $i = 1, \ldots, r_0$, to define the LIS basis $\basis_{0, r_0} = [\lis_{0,1} , \lis_{0,2} \ldots, \lis_{0,r_0} ]$, which spans an $r_0$-dimensional subspace in $\mathbb{R}^R_0$. 
\end{definition}

The eigenvalues in \eqref{eq:exp_H0} provide empirical sensitivity
measures of the likelihood function relative to the prior (which here
is i.i.d.~Gaussian)  along corresponding eigenvectors \cite{DimRedu:Cui_etal_2014,DimRedu:Zahm_etal_2018}.
Eigenvectors corresponding to eigenvalues less than $1$ can be interpreted as parameter directions where the likelihood is dominated by the prior.
Thus, we typically choose a value less than one for the truncation threshold, i.e., $\eigentrunc < 1$. 

\subsection{LIS enrichment} \label{sec:lis_enrich}
Because the computational cost of a matrix vector product with the Gauss--Newton Hessian scales at least linearly with the degrees of freedom $M_\ell$ of the forward model on level $\ell$, constructing the LIS can be computationally costly.
We present a new approach to accelerate the LIS construction by employing a {\it recursive LIS enrichment} using the hierarchy of forward models and parameter discretisations. 
The resulting hierarchy of LISs will be used to reduce the computational complexity of constructing and operating with the resulting DILI proposals.

We reuse the LIS bases computed on the coarser levels by 'lifting'
them and then recursively
enrich them at each new level using a {\it Rayleigh-Ritz procedure},
rather than recomputing the entire basis from scratch on each level.
Ideally, the subspace added on each level will have decreasing
dimension, as the model and parameter approximations were assumed to
converge with $\ell \to \infty$ and thus no longer provide additional
information for the parameter inference. 
\begin{definition}[Lifted LIS basis]
Suppose we have an orthogonal LIS basis $\basis_{\mell, r} \in
\R^{R_{\mell} \times r_{\mell}}$ on level $\mell$. 
We lift $\basis_{\mell, r}$ from the coarse parameter space $\R^{R_{\mell}}$ to the fine parameter space $\R^{R_\ell}$ using the basis matrix $\Theta_{\ell, c}$ defined in Section \ref{sec:notation}.
The lifted LIS basis vectors are collected in the matrix
\begin{equation}
\basis_{\ell, c} = \Theta_{\ell, c} \, \basis_{\mell, r}.
\end{equation}
\end{definition}

\begin{proposition}
The lifted LIS basis matrix $\basis_{\ell, c}$ has orthonormal columns
that span an $r_{\mell}$-dimensional subspace in $\R^{R_\ell}$, i.e.,
$\basis_{\ell, c}^\top \basis_{\ell, c}^{} = \cI_{r_{\mell}}$.
\end{proposition}

\begin{proof}
The proof directly follows as the matrix $\Theta_{\ell, c}$ has orthonormal columns. 
\end{proof}

Given $K_\ell$ samples $\{\bv_\ell^{(k)}\}_{k=1}^{K_\ell}$ from the
reference distribution $p_\ell^\ast(\cdot)$, let
$\widehat\hessian_\ell$ be the resulting sample-averaged Gauss--Newton Hessian.
To enrich the lifted LIS basis $\basis_{\ell, c}$ we now identify
likelihood-sensitive parameter directions in the null space 
$\textrm{null}(\basis_{\ell, c})$ by recursively optimising the
Rayleigh quotient in the orthogonal complement of
$\range(\basis_{\ell, c})$, i.e., 
\begin{align}\label{eq:Rayleigh_l}
\;\;& \;\lis_{\ell,k+1} = \argmax_{\| \boldsymbol \phi \| = 1} \langle \boldsymbol \phi, \widehat\hessian_\ell \, \boldsymbol \phi \rangle,  \\
\textrm{subject\;to\;\;} & \Pi_{\ell, c} \phi = 0 \textrm{\;\,and\;\,} \langle \boldsymbol \phi, \lis_{\ell,i} \rangle = 0, \textrm{\;for\;} i = 1,  \ldots, k, \nonumber
\end{align}
where $\Pi_{\ell, c} = \basis_{\ell, c}^{} \basis_{\ell, c}^\top$ is an orthogonal projector.
This optimisation problem can be solved as an eigenvalue problem using the {\it Rayleigh-Ritz procedure} \cite{Lin:Saad_2011}.

\begin{theorem}\label{theo:Rayleigh-Ritz}
The optimisation problem \eqref{eq:Rayleigh_l} is equivalent to
finding the leading eigenvectors of the projected eigenproblem
\begin{equation}
\big(\cI_{R_\ell} - \Pi_{\ell, c} \big) \, \widehat\hessian_\ell\, \lis_{\ell,i} = \gamma_{\ell,i} \lis_{\ell,i}, \quad \|\lis_{\ell,i}\| = 1.
\label{eq:def_exp_Hl}
\end{equation}
\end{theorem}
\begin{proof}
This result follows from the properties of orthogonal projectors and
of the stationary points of the Rayleigh quotient. Here, we sketch the proof as follows.
The constraint $\Pi_{\ell, c} \boldsymbol \phi = 0$ implies
$\boldsymbol \phi = (\cI_{R_\ell} - \Pi_{\ell, c} ) \boldsymbol \phi$,
since $(\cI_{R_\ell} - \Pi_{\ell, c} )$ is also an orthogonal projector.
Hence, the optimisation problem becomes
\[
\lis_{\ell,k+1} = \argmax_{\| \boldsymbol \phi \| = 1} \langle
\boldsymbol \phi, (\cI_{R_\ell} - \Pi_{\ell, c} )
\widehat\hessian_\ell (\cI_{R_\ell} - \Pi_{\ell, c} ) \, \boldsymbol
\phi \rangle,  \ \ 
\text{subject to} \ \ \langle \boldsymbol \phi, \lis_{\ell,i} \rangle
= 0, \  i = 1,  \ldots, k.
\]
The solutions (for $k = 1, 2, \ldots$) to these optimisation problems are given by the leading eigenvectors of the eigenproblem
\[
\big(\cI_{R_\ell} - \Pi_{\ell, c} \big) \, \widehat\hessian_\ell\, \big(\cI_{R_\ell} - \Pi_{\ell, c} \big)\, \lis_{\ell,i} = \gamma_{\ell,i} \lis_{\ell,i}.
\]
However, since $\lis_{\ell,i} \in \textrm{range} \big(\cI_{R_\ell} -
\Pi_{\ell, c} \big) $ this is equivalent to 
\[
\big(\cI_{R_\ell} - \Pi_{\ell, c} \big) \, \widehat\hessian_\ell\, \lis_{\ell,i} = \gamma_{\ell,i} \lis_{\ell,i}.
\]
\end{proof}

\begin{definition}[LIS enrichment on level $\ell$]\label{def:deflation}
The leading $s_\ell$ (normalised) eigenvectors of the eigenproblem \eqref{eq:def_exp_Hl} with eigenvalues $\gamma_{\ell,i} > \eigentrunc$ are denoted by 
\begin{equation}\label{eq:auxiliary_lis}
\basis_{\ell, f} = [\lis_{\ell,1}, \dots, \lis_{\ell,s_\ell}].
\end{equation}
They are added to the lifted LIS basis from
level $\mell$ to form the enriched LIS basis
\begin{equation}
\basis_{\ell, r} = [\basis_{\ell, c}, \basis_{\ell, f}]
\end{equation}
 on level~$\ell$, where the basis vectors in \eqref{eq:auxiliary_lis} denote the auxiliary ``fine scale'' directions added on level $\ell$.
By construction, all the LIS basis vectors at level $\ell$ are mutually orthogonal. That is, $\basis_{\ell, r}^\top \basis_{\ell, r}^{} = \cI_{r_\ell}$.
We also have $r_\ell = r_{\mell} + s_\ell$.
\end{definition}

By construction, the LIS basis $\basis_{\ell, r}$ is block upper triangular and can be recursively defined as
\begin{equation}\label{eq:upper_tri_lis}
\basis_{\ell, r} = [\basis_{\ell, c}, \basis_{\ell, f}] = \begin{bmatrix}
\basis_{\mell,r} &\cZ_{\ell,c}\\
0  &\cZ_{\ell,f}
\end{bmatrix},
\end{equation}
where $\cZ_{\ell,c} = \Theta_{\ell, c}^\top \basis_{\ell, f} \in \R^{ R_{\mell} \times  s_{\ell} }$, $\cZ_{\ell,f} = \Theta_{\ell, f}^\top \basis_{\ell, f} \in \R^{ (R_{\ell} - R_{\mell}) \times  s_{\ell} }$, and $\basis_{\mell, r} \in \R^{ R_{\mell} \times r_{\mell}}$.
We have $s_\ell = r_{\ell} - r_{\mell}$ and define  $s_0 = r_0$ for consistency. 
The hierarchical LIS reduces the computational cost of operating with the LIS basis and the associated storage cost. This is critical for building efficient multilevel DILI proposals that will be discussed later.
In addition, the recursive LIS enrichment is computationally more efficient, since the amount of costly PDE solves on the finer levels will be significantly reduced. In Appendix \ref{sec:LIS_cost}, we develop heuristics to demonstrate the reduction factors of the hierarchical construction of LIS basis in terms of the storage and the number of matrix vector products. 



\section{Multilevel DILI MCMC}\label{sec_MLDILI}

To compute the multilevel MCMC estimator, we need to construct Markov chains $\smash{\{\rvml^{(\ell, j)} \}}$ and $\{\rvl^{(\ell, j)} \}$ for adjacent levels $\mell$ and $\ell$ with invariant densities $\smash{\pi_{\mell}(\bv_{\mell} | \data)}$ and $\pi_{\ell}(\bv_{\ell}| \data)$, respectively. As discussed in Remark \ref{remark:var_cov}, it is crucial that the QoIs produced by the two Markov chains $\smash{\{\rvml^{(\ell, j)} \}}$ and $\{\rvl^{(\ell, j)} \}$ are positively correlated, i.e., $\textrm{Cov}_{\Delta_{\ell,\mell}}(Q_\ell, Q_{\mell}) > 0$, so that the within-level variance
$\Var_{\Delta_{\ell,\mell}}(D_\ell)$ is reduced.
\blue{Here, we design a computationally efficient way in Section \ref{sec:coupling} to couple DILI proposals within the original MLMCMC \cite{MCMC:KST_2013}, we introduce the computational framework in Section \ref{sec:dili}, and then provide an alternative sampling strategy in Section \ref{sec:pooling} that is more suitable for a parallel implementation.}

\subsection{Coupled DILI proposal}
\label{sec:coupling}

Let $\rvml^{(\ell,j)} = \bv_{\mell}^\ast$ and $\rvl^{(\ell,j)} = \bv_\ell^\ast$ be the $j$-th states of the Markov chains at levels $\mell$ and $\ell$, respectively. The state at level $\ell$ has the form ${\bv_\ell^\ast = (\bv_{\ell, c}^\ast, \bv_{\ell, f}^\ast)}$, corresponding to the coarse part of the parameters (shared with level $\mell$) and the refined part, respectively. 
The two Markov chains are called {\it coupled} at the $j$-th state if $\bv_{\ell, c}^\ast = \bv_{\mell}^\ast$.
Thus, assuming the two chains to be coupled at the $j$th state, we first present the general form of the multilevel MCMC for generating the next pair of coupled states, and then design the hierarchical DILI proposal within this general framework.

Following \cite{MCMC:KST_2013}, we assume that we can generate independent posterior samples $\mathcal{V}_{\mell} = \{\bv_{\mell}^{(i)} \}_{i =1}^{N_{\ell}}$ on level $\mell$. In practice, this is achieved (approximatively) by sub-sampling a Markov chain that targets the level $\mell$ posterior with a sub-sampling rate that depends on the sample autocorrelation \cite[Sect.~3]{MCMC:KST_2013}.
In other words, coupled posterior samples from $\pi(\bv_{\mell} |
\by)$ and $\pi(\bv_{\ell} | \by)$ are generated by using the posterior
$\pi(\bv_{\mell} | \by)$ on level $\mell$ as the proposal distribution
for the Markov chain on level $\ell$, thus reducing the within-level
variance $\Var_{\Delta_{\ell,\mell}}(D_\ell)$.

The proposed candidate $\bv'_{\mell}  \sim \pi_{\mell}(\cdot | \by)$
is assumed to be independent of the current state $\bv_{\mell}^\ast$.
To sample from the refined posterior  $\pi_{\ell}(\bv_{\ell}| \data)$,
we then consider the factorised proposal
\begin{align}\label{eq:coupling_proposal}
q \big( \bv_{\ell}' \,\big|\, \bv_{\ell}^\ast \big) = q \big( \bv_{\ell,c}'\,, \,\bv_{\ell,f}' \,\big|\, \bv_{\ell}^\ast \big) = \pi_{\mell}^{} \big(\bv_{\ell, c}' | \by \big) \, q \big( \bv_{\ell,f}' \,|\, \bv_{\ell}^\ast, \bv_{\ell,c}' \big) ,
\end{align}
where the coarse part $\bv_{\ell, c}'$ of the proposal is set to be
the (independent) proposal $\bv_{\mell}'$ from level $\mell$.
The proposal candidate $\bv'_\ell$ conditioned on $\bv_\ell^\ast$ can then be expressed as
\begin{align}\label{eq:proposal_1}
& \bv'_{\ell,c}  = \bv'_{\mell} & ( \textrm{copy from level \;} \mell \textrm{\;proposal}) ,\\
& \bv_{\ell,f}' \sim q \big( \, \cdot \,|\, \bv_{\ell}^\ast, \bv_{\ell,c}' \big) & (\textrm{conditional proposal}).\label{eq:proposal_2}
\end{align}
Based on the factorised proposal \eqref{eq:coupling_proposal}, the acceptance probability for the chain targeting the level $\ell$ posterior $\pi^{}_\ell\big(\bv^{}_{\ell} \,|\,
\data\big)$ is of the form
\begin{align}\label{eq:accept_multi}
\alpha^{\rm ML}_\ell(\bv_{\ell}^\ast, \bv_{\ell}'  ) = \min\left\{1, \frac{\pi_\ell\big(\bv_{\ell}' | \data\big)\,\pi_{\mell}\big(\bv_{\mell}^\ast| \data\big)}{\pi_\ell\big(\bv_{\ell}^\ast | \data\big)\,\pi_{\mell}\big(\bv_{\mell}^\prime| \data\big)}
\frac{q \big(  \bv_{\ell, f}^\ast | \bv_{\ell}^\prime, \bv_{\mell}^\ast \big)}{q \big( \bv_{\ell,f}^\prime | \bv_{\ell}^\ast,\bv_{\mell}^\prime \big)}  \right\}.
\end{align}

\begin{figure}[ht]
\centering
\begin{tikzpicture}[scale=0.91]
\tikzset{point1/.style = {draw, circle,  fill = black, inner sep = 2pt}}
\tikzset{point2/.style = {draw, circle, very  thick, fill = white, inner sep = 2pt}}
\node (C0) at (-2,0) []  {} ;
\node (C1) at (0,0) [point1, label={above:$\bV_{\!\!\mell}^{(\ell, j)} = \bv_{\mell}^\ast$}]  {} ;
\node (C2) at (5,0) [point1, label={above:$\bV_{\!\!\mell}^{(\ell, j+1)} = \bv'_{\mell}$}]  {} ;
\node (C3) at (10,0) [point1, label={above:$\bV_{\!\!\mell}^{(\ell, j+2)} = \bv_{\mell}^\circ$}]  {} ;
\node (C4) at (12,0)  []  {} ;
\node (PC2) at (3,-2) [point2, label={left:$\bv'_{\mell}$\;\,}]  {} ;
\node (PC3) at (7,-2) [point2, label={left:$\bv_{\mell}^\circ$\;\,}]  {} ;
\node (PF2) at (3,-4) [point2, label={left:$(\bv'_{\ell, c}, \bv'_{\ell, f})$}]  {} ;
\node (PF3) at (7,-4) [point2, label={left:$(\bv_{\ell, c}^\circ, \bv_{\ell, f}^\circ)$}]  {} ;
\node (F0) at (-2,-6) []  {} ;
\node (F1) at (0,-6) [point1, label={below:$\bV_{\!\!\ell}^{(\ell, j)} = (\bv_{\ell, c}^\ast, \bv_{\ell, f}^\ast)$}]  {} ;
\node (F2) at (5,-6) [point1, label={below:$\bV_{\!\!\ell}^{(\ell, j+1)} = (\bv_{\ell, c}^\ast, \bv_{\ell, f}^\ast)$}]  {} ;
\node (F3) at (10,-6) [point1, label={below:$\bV_{\!\!\ell}^{(\ell, j+2)} = (\bv_{\ell, c}^\circ, \bv_{\ell, f}^\circ)$}]  {} ;
\node (F4) at (12,-6)  []  {} ;
\draw[->, very thick, dashed] (C0) -- (C1);
\draw[->, very thick] (C1) -- (C2);
\draw[->, very thick] (C2) -- (C3);
\draw[->, very thick, dashed] (C3) -- (C4);
\draw[->, very thick, dashed] (F0) -- (F1);
\draw[->, very thick] (F1) -- (F2);
\draw[->, very thick] (F2) -- (F3);
\draw[->, very thick, dashed] (F3) -- (F4);
\draw[<->,double, blue] (C1) -- node[left] {$\bv_{\ell, c}^\ast = \bv_{\mell}^\ast$} (F1) ;
\draw[<->,double, blue] (C3) -- node[right] {$\bv_{\ell, c}^\circ = \bv_{\mell}^\circ$} (F3) ;
\draw[<->,double, blue] (PC2) -- node[left] {$\bv'_{\ell, c} = \bv'_{\mell}$} (PF2) ;
\draw[<->,double, blue] (PC3) -- node[left] {$\bv_{\ell, c}^\circ = \bv_{\mell}^\circ$} (PF3) ;
\draw[->, dashed, very thick, teal] (C1) -- node[below, sloped] {propose} (PC2);
\draw[->, dashed, very thick, teal] (PC2) -- node[above, sloped] {accept} (C2);
\draw[->, dashed, very thick, teal] (C2) -- node[below, sloped] {propose} (PC3);
\draw[->, dashed, very thick, teal] (PC3) -- node[above, sloped] {accept} (C3);
\draw[->, dashed, very thick, teal] (F1) -- node[above, sloped] {propose} (PF2);
\draw[->, dashed, very thick, red!50] (PF2) -- node[below, sloped] {reject} (F2);
\draw[->, dashed, very thick, teal] (F2) -- node[above, sloped] {propose} (PF3);
\draw[->, dashed, very thick, teal] (PF3) -- node[below, sloped] {accept} (F3);
\end{tikzpicture}

\caption{ \label{fig:couping} \blue{This diagram illustrates the coupling strategy with double arrows representing the coupling of two MCMC states as well as the coupling of two proposal candidates across levels. The dashed arrows represent the proposal and the accept/reject steps. At each iteration, the proposal candidates are coupled by construction, while the samples $\bv_{\ell}^{(\ell, j)}$ and $\bv_{\mell}^{(\ell, j)}$ on  two adjacent levels are only coupled whenever the proposal on level $\ell$ is accepted, i.e., for all cases $j, j+1, j+2$ here.}
}
\end{figure}

Figure \ref{fig:couping} shows a schematic of the coupling
strategy. The double arrows represent the coupling of the two MCMC states,
as well as the coupling of the two proposal candidates across
levels. The dashed arrows represent the proposal and
acceptance/rejection steps.
The top half represents the Markov chain on level $\mell$.
The bottom half represents the Markov chain on level $\ell$. Since all the proposal candidates are coupled, all states that follow the acceptance of a proposal candidate on level $\ell$ are also coupled with the
corresponding state on level $\mell$.

\subsubsection{DILI proposal}
Then, we design the DILI proposal using the hierarchical LIS introduced in Section \ref{sec_MLLIS}.
Recall that the discretised DILI proposal \eqref{eq:operator_weighted_v} is
\begin{equation}
\label{eq:operator_weighted_l}
\bv'_\ell =  \cA_\ell^{} \bv_\ell^\ast + \cB_\ell^{} \bxi_\ell, \quad \textrm{where\quad} \bxi_\ell \sim \normal\big(0, \cI_{R_\ell}\big),
\end{equation}
as it was introduced in \cite{MCMC:CLM_2016}.
Suppose we have a LIS basis $\basis_{\ell, r} \in \R^{R_\ell \times r_\ell}$. 
By treating the likelihood-informed parameter directions and the prior-dominated directions separately, we can construct the matrices $\cA_\ell$ and $\cB_\ell$ as
\begin{align}
\label{eq:operators_al}
\cA_\ell & = \basis_{\ell,r} \, \cA_{\ell,r} \, {\basis_{\ell,r}^\top} + a_\perp ( \cI_{R_\ell} - \Pi_{\ell})   \in \R^{R_
\ell\times R_\ell},\\
\cB_\ell^2 & = \basis_{\ell,r} \, \cB_{\ell,r}^2 \, {\basis_{\ell,r}^\top} + b_\perp^{2} ( \cI_{R_\ell} - \Pi_{\ell})\in \R^{R_
	\ell\times R_\ell},
\label{eq:operators_bl}
\end{align}
where $A_{\ell,r},B_{\ell,r}\in\R^{r_
	\ell\times r_\ell}$, $a_\perp$ and $b_\perp\in\R$ and
      $\Pi_{\ell} = \basis_{\ell, r}^{} \basis_{\ell, r}^\top$ are
      rank-$r_\ell$ orthogonal projectors. 

\begin{corollary}\label{coro:accept_DILI}
In the proposal \eqref{eq:operator_weighted_l}, suppose that 
$\cA_{\ell, r}, \cB_{\ell, r}\in \R^{r_\ell \times r_\ell}$ are
non-singular matrices satisfying $\cA_{\ell,r}^2 + \cB_{\ell,r}^2 =
\cI_{\ell,r}$, and $a_\perp$ and $b_\perp$ are scalars 
satisfying  $a_\perp^2 + b_\perp^2 = 1$.
Then, the corresponding proposal distribution
$q(\bv_{\ell}'|\bv_{\ell}^*)$ satisfies the  conditions of Theorem
\ref{theo:1} and has the prior as its invariant measure, i.e., this
proposal  has acceptance probability one if we use it to sample the
prior.  The acceptance probability as samples from
$\pi_\ell\big(\bv_{\ell} | \data\big)$ is
\begin{equation}\label{eq:accept_dili}
\alpha \big(\bv_\ell^\ast, \bv^\prime_\ell\big) = \min \Big\{ 1, \exp\Big[ \potential_\ell^{} \big(\bv_\ell^\ast; \data \big) - \potential_\ell^{} \big(\bv^\prime_\ell; \data \big)  \Big] \Big\}.
\end{equation}
\end{corollary}
\begin{proof}
Given $\cA_{\ell,r}^2 + \cB_{\ell,r}^2 = \cI_{\ell,r}$, the symmetric matrices $\cA_{\ell, r}$ and $\cB_{\ell, r}$ can be simultaneously diagonalised under some orthogonal transformation. 
Thus, the operators $\cA_\ell$ and $\cB_\ell$ can be simultaneously diagonalised, where the eigenspectrum of $\cA_\ell$ consists of the eigenvalues of $\cA_{\ell, r}$ and $a_\perp$, and the same applies to $\cB_\ell$.
This way, it is easy to check that the proposal distribution
$q(\bv_{\ell}'|\bv_{\ell}^*)$ has the prior as invariant measure and
that the conditions of Theorem~\ref{theo:1} are satisfied. The form of
the acceptance probability to sample from $\pi_\ell\big(\bv_{\ell} |
\data\big)$ directly follows from the acceptance
probability defined in Theorem \ref{theo:1}.
\end{proof}

We use the empirical posterior covariance, commonly used in adaptive MCMC \cite{MCMC:RoRo_1998,MCMC:HST_2001,IP:Haario_etal_2004} to construct matrices $\cA_{\ell, r}$ and $\cB_{\ell, r}$ for our DILI proposal \eqref{eq:operator_weighted_l}.
On each level, the empirical covariance matrix $\Sigma_{\ell, r}\in \R^{r_\ell \times
  r_\ell}$ is estimated from past posterior samples projected onto the
LIS. Given a jump size $\Delta t$, we can then define the matrices $\cA_{\ell, r}^{}$ and $\cB_{\ell, r}^{2}$ by
\begin{align*}
\textstyle  \cA_{\ell, r} & = \textstyle (2 \cI_{r_\ell} \!+\! \Delta t \Sigma_{\ell, r})^{-1}(2 \cI_{r_\ell} \!-\! \Delta t \Sigma_{\ell, r}) = \cI_{r_\ell}  - 2 \big(  \cI_{r_\ell} \!+\! \frac{\Delta t}{2}\Sigma_{\ell, r} \!\big)^{\!-1} \big(\frac{\Delta t}{2}\Sigma_{\ell, r} \big), \\
\cB_{\ell, r}^2 & = \textstyle \cI_{r_\ell}^{} - \cA_{\ell, r}^2  = 4 \, \big( 2\,\cI_{r_\ell} + \big(\frac{\Delta t}{2} \Sigma_{\ell, r} \big)^{-1} + \frac{\Delta t}{2} \Sigma_{\ell, r} \big)^{-1},
\end{align*}
respectively. The operators $\cA_{\ell, r}$ and $\cB_{\ell, r}$
satisfy $\cA_{\ell, r}^2 + \cB_{\ell, r}^2 = \cI_{\ell, r}^{}$ by
construction.

By estimating the empirical covariance within the subspace, common conditions such as the diminishing adaptation \cite{MCMC:AnMou_2006,MCMC:RoRo_2007} for the convergence of adaptive MCMC can be easily satisfied. In addition, we adopt a finite adaptation strategy in our numerical implementation, in which only the samples generated post adaptation are used for estimating QoIs.

\subsubsection{Conditional DILI proposal}
On level $0$, the vanilla DILI proposal (cf.~\cite{MCMC:CLM_2016}) can be used to sample the Markov chain with invariant distribution $\pi_{0}(\bv_{0}| \data)$. 
On level $\ell$, to simulate coupled Markov chains using the proposal mechanism defined in \eqref{eq:coupling_proposal}--\eqref{eq:proposal_2}, a key step is to use DILI to generate the fine components $\bv_{\ell, f}'$ of
the proposal candidate and thus to fix the conditional probability 
$q(\bv_{\ell,f}' | \bv_{\ell}^\ast, \bv_{\ell,c}' )$. 
Defining the precision matrix
\begin{equation}\label{eq:Ql}
\cP_\ell^{} = \cB_\ell^{-2} = \basis_{\ell, r}^{} \, \cB_{\ell,r}^{-2} \,\basis_{\ell, r}^\top + b_\perp^{-2} ( \cI_{R_\ell}^{} - \Pi_{\ell}^{}),
\end{equation}
the DILI proposal \eqref{eq:operator_weighted_l} can be split as follows:
\begin{equation}\label{eq:DILIproposal_split}
\begin{bmatrix} \bv'_{\ell, c} \\ \bv'_{\ell, f} \end{bmatrix} = \cA_\ell^{} \bv_\ell^\ast  + \begin{bmatrix} \br_{\ell, c}  \\ \br_{\ell, f} \end{bmatrix}, 
\quad \begin{bmatrix} \br_{\ell, c}  \\ \br_{\ell, f} \end{bmatrix} 
\sim \normal\Big(0, \begin{bmatrix}
\cP_{\ell,cc} &\cP_{\ell,cf}\\
\cP_{\ell,fc} &\cP_{\ell,\ff}
\end{bmatrix}^{-1}\Big),
\end{equation}
where the partitions of the vectors and of the matrix $\cP_\ell$
correspond to the parameter coordinates shared with level $\mell$ and
the refined parameter coordinates on level $\ell$.

\blue{To draw candidate samples from the factorised
proposal distribution $\pi_{\mell}^{} \big(\bv_{\ell, c}' | \by \big)\,
q ( \bv_{\ell,f}' | \bv_{\ell}^\ast, \bv_{\ell,c}' )$ defined in \eqref{eq:coupling_proposal} we use the procedure outlined in Algorithm \ref{alg:conditional_dili}, which employs the
DILI proposal in the form of \eqref{eq:DILIproposal_split} for
the conditional distribution $q( \bv_{\ell,f}' | \bv_{\ell}^\ast, \bv_{\ell,c}' )$.}

\begin{algorithm}[th]
  \caption{Conditional DILI proposal.}
  \label{def:conditional_proposal}
\blue{
  \textbf{Input:} A proposal $\bv'_{\ell,c}$ drawn from $\pi_{\mell}^{} (\bv_{\ell, c}' | \by )$ using a sub-sampled Markov chain.\\[0.5ex]
  \textbf{Output:} A joint, candidate proposal $\bv_\ell^\prime = (\bv_{\ell,c}^\prime, \bv_{\ell,f}^\prime)$ on the fine level based on \eqref{eq:DILIproposal_split}. 
\begin{algorithmic}[1]
\Procedure{Conditional DILI proposal\label{alg:conditional_dili}}{}
	\State With $\bv'_{\ell,c}$ and $\bv_{\ell}^\ast$ known, compute the `residual' $\br_{\ell, c}= \bv'_{\ell, c} - \Theta_{\ell, c}^\top\,\cA_\ell^{} \, \bv_\ell^\ast$ using \eqref{eq:DILIproposal_split}.\!\! 
\State Draw a random variable $\br_{\ell, f}$
          conditioned on $\br_{\ell, c}$ such that jointly
          $(\br_{\ell, c}, \br_{\ell, f}) \sim \normal(0, \cP_\ell^{-1})$.\linebreak
          \hspace*{0.4cm} Due to
          \eqref{eq:DILIproposal_split}, the fine-level components of the
          proposed  candidate $\bv'_{\ell,f}$ then satisfy 
\begin{equation}\label{eq:conditional_gaussian}
\bv'_{\ell,f} = \Theta_{\ell, f}^\top\, \cA_\ell^{}\, \bv_\ell^\ast +
\br_{\ell, f}^{} , \quad \br_{\ell, f}^{} 
\sim \cN\big(\!-\!\cP_{\ell,\ff}^{-1}\cP_{\ell, fc}^{}\br_{\ell,c}^{},
\cP_{\ell,\ff}^{-1} \big) .  
\end{equation}
\EndProcedure
\end{algorithmic}
%
}
\end{algorithm}

\begin{corollary}\label{coro:conditional}
Using the above procedure to draw candidates from the factorised
proposal distribution $\pi_{\mell}^{} \big(\bv_{\ell, c}' | \by \big) q \big( \bv_{\ell,f}' | \bv_{\ell}^\ast, \bv_{\ell,c}' \big)$,
the acceptance probability to sample from the posterior distribution
$\pi_\ell(\bv_{\ell} | \data)$ is
\begin{equation*}
\alpha^{\rm ML}_\ell \big(\bv_\ell^\ast, \bv^\prime_\ell\big) = \min \left\{ 1, \exp\!\left[ \Big(\potential_\ell^{} \big(\bv_\ell^\ast; \data\big) \!-\! \potential_{\mell}^{} \big(\bv_{\mell}^\ast; \data\big) \Big) \!-\! \Big(\potential_\ell^{} \big(\bv_\ell'; \data\big) \!-\! \potential_{\mell}^{} \big(\bv_{\mell}'; \data\big) \Big)  \right] \right\}.
\end{equation*}
\end{corollary}
\begin{proof}
See Appendix \ref{sec:proof_conditional}.
\end{proof}

\subsubsection{Generating conditional samples}\label{sec:DILI_contitional}
The computational cost of the coupling procedure is dictated by the multiplication with $\cA_\ell$ in Step 2 and the generation of conditional proposal samples in Step 3. 
The multiplication with $\cA_\ell$ has a computational complexity of $\textstyle \cO( \sum_{j = 0}^{\ell} R_{j} s_{j})$ using the low-rank representation \eqref{eq:operators_al} and the upper-triangular hierarchical LIS basis in \eqref{eq:upper_tri_lis}, which has the form 
\[
\basis_{\ell, r} = [\basis_{\ell, c}, \basis_{\ell, f}] = \begin{bmatrix}
\basis_{\mell,r} &\cZ_{\ell,c}\\
0  &\cZ_{\ell,f}
\end{bmatrix}.
\]

We can also exploit the hierarchical LIS to reduce the computational
cost of generating conditional proposal samples.
As shown in Equation \eqref{eq:Ql}, given the LIS basis $\basis_{\ell,r}$, the precision matrix $\cP_\ell$ is dictated by the matrix $\cB_{\ell,r}^{-2}$, which has the block form
\begin{equation}
\cB_{\ell,r}^{-2} = \begin{bmatrix}
\Xi_{\ell,cc} &\Xi_{\ell,cf}\\
\Xi_{\ell,fc} &\Xi_{\ell,\ff}
\end{bmatrix},
\end{equation}
corresponding to the splitting of the enriched LIS basis into $\basis_{\ell, c}$ and $\basis_{\ell, f}$. 
Generating conditional proposal samples only involves the blocks
$\cP_{\ell,\ff}$ and $\cP_{\ell, fc}$ in the  matrix $\cP_\ell$, i.e.,
\begin{align}
\cP_{\ell,\ff} &= \cZ_{\ell,f} \, \big(\Xi_{\ell,\ff} - b_\perp^{-2}\,\cI \,\big) \, \cZ_{\ell,f}^\top  + b_\perp^{-2}\,\cI_{\ell,f}\,,\\
\cP_{\ell,fc} &=\cZ_{\ell,f} \Xi_{\ell,fc}  \basis_{\mell,r}^\top + \cZ_{\ell,f} \Xi_{\ell,\ff} \cZ_{\ell,c}^\top - b_\perp^{-2} \cZ_{\ell,f}\cZ_{\ell,c}^\top\,,
\end{align}
which in turn only require the blocks $\Xi_{\ell,fc} \in  \R^{ s_{\ell} \times r_{\mell}}$ and $\Xi_{\ell,\ff} \in \R^{ s_{\ell}  \times s_{\ell} }$ in the matrix $\cB_{\ell,r}^{-2}$.

We derive low-rank operations to avoid the direct inversion or
factorisation of the matrices $\cP_{\ell,\ff}$ and $\cP_{\ell, fc}$ in
the generation of conditional samples and to reduce the computational cost.
Suppose the block $\cZ_{\ell,f} \in \R^{ (R_\ell - R_{\mell}) \times s_\ell}$ has the thin QR factorisation 
\begin{equation}\label{eq:qr_cond}
\cZ_{\ell,f} = \cU_{\ell} \cT_\ell,
\end{equation}
where $\cU_\ell$ has orthonormal columns and $\cT_\ell$ is upper triangular.
Then the matrix $\cP_{\ell,\ff}$ can be expressed as
\[
\cP_{\ell,\ff} = b_\perp^{-2} \Big ( \cU_{\ell} \big( \cT_\ell (b_\perp^{2}\,\Xi_{\ell,\ff} - \cI ) \cT_\ell^\top \big) \cU_{\ell}^\top + \,\cI_{\ell,f} \Big).
\]
Computing the $s_{\ell} \times s_{\ell}$ eigendecomposition 
\begin{equation}\label{eq:eig_cond}
\cT_\ell (b_\perp^{2}\,\Xi_{\ell,\ff} - \cI ) \cT_\ell^\top = \cW_\ell \cD_\ell \cW_\ell^\top,
\end{equation}
where $\cW_\ell$ and $\cD_\ell$ are respectively orthogonal and diagonal matrices, we have
\[
\cP_{\ell,\ff} = b_\perp^{-2} \Big ( \Phi_\ell \, \cD_\ell \,
\Phi_\ell^\top + \,\cI_{\ell,f} \Big), \quad \text{with} \quad \Phi_\ell := \cU_{\ell} \cW_\ell .
\]
Note that $\Phi_\ell \in \R^{ (R_\ell - R_{\mell}) \times s_\ell}$ has
orthonormal columns, so that
\begin{align}
\cP_{\ell,\ff}^{-1} \cP_{\ell,fc} & = b_\perp^2 \, \Phi_\ell \underbrace{\big( (\cD_\ell + \cI)^{-1} \cW_\ell^\top \cT_\ell \big)}_{s_{\ell} \times s_{\ell} } \underbrace{\big( \Xi_{\ell,fc}  \basis_{\mell,r}^\top + \Xi_{\ell,\ff} \cZ_{\ell,c}^\top - b_\perp^{-2} \cZ_{\ell,c}^\top \big)}_{s_{\ell}  \times R_{\mell} }, \label{eq:PinvP}\\
\cP_{\ell,\ff}^{-\frac12} & = b_\perp \, \Big( \Phi_\ell \underbrace{ \big( (\cD_\ell + \cI)^{-\frac12} - \cI \big)}_{ s_{\ell} \times s_{\ell} } \Phi_\ell^\top + \cI_{\ell,f} \Big) . \label{eq:Phalf}
\end{align}
Using these representations of the matrices $\smash{
  \cP_{\ell,\ff}^{-1} \cP_{\ell,fc} }$ and $\smash{
  \cP_{\ell,\ff}^{-1/2}}$, the conditional Gaussian in
\eqref{eq:conditional_gaussian}  can be simulated efficiently using
\begin{equation}
\br_{\ell, f}^{} | \br_{\ell, c}^{} = - \cP_{\ell,\ff}^{-1}\,\cP_{\ell, fc}^{}\,\br_{\ell,c}^{} + \cP_{\ell,\ff}^{-\frac12} \xi, \textrm{\quad where\quad}\xi \sim \normal\big(0,\, \cI_{(R_{\ell} - R_{\mell})}\big).
\end{equation}
The associated computational cost is $\cO(R_{\ell} s_{\ell} )$.
%


\subsection{Final MLDILI algorithm}\label{sec:dili}

Here, we assemble all the elements of the multilevel DILI method
defined in the previous sections in algorithmic form.
For the base level ($\ell = 0$ ), the LIS construction and the DILI--MCMC sampling are presented in Algorithm \ref{algo:level_0}. The recursive LIS construction and the coupled DILI--MCMC are presented in Algorithm \ref{algo:level_l}.

\renewcommand{\algorithmicensure}{\textbf{Note:}}
\begin{algorithm}[t]
    \textbf{Input:} A set of samples $\mathcal{W}_0 =\{ \bv_0^{(k)} \}_{k = 1}^{K_0}$ drawn from the base level reference $p_0^\ast(\cdot)$, the number of MCMC iterations $N_0$, and an initial MCMC state $\bV_0^{(0)}$.\\
    \textbf{Output:} A LIS basis $\basis_{0, r}$ and a Markov chain of posterior samples $\mathcal{V}_{0} = \{ \bV_0^{(j)} \}_{j = 1}^{N_0}$.
  \begin{algorithmic}[1]
	\Procedure{Base level LIS and MCMC}{}
	\State Use $\mathcal{W}_0$ to solve the eigenproblem in \eqref{eq:exp_H0} to obtain the base level LIS basis $\basis_{0, r}$.
	\State Estimate the empirical covariance matrix $\Sigma_{0, r}$ from the samples in $\mathcal{W}_0$ and define the operators $\cA_0$ and $\cB_0$ as in \eqref{eq:operators_al}--\eqref{eq:operators_bl}.
	\For {$j = 1, \ldots, N_0$}
	\State Propose a candidate $\bv_0^\prime$ using the base level proposal in \eqref{eq:operator_weighted_l}.
	\State Compute the acceptance probability $\alpha(\bV_0^{(j -1 )},\bv_0^\prime)$ defined in \eqref{eq:accept_dili}.
	\State With probability $\alpha(\bV_0^{(j -1 )},\bv_0^\prime)$, set $\bV_0^{(j)} = \bv_0^\prime$, otherwise set $\bV_0^{(j)} = \bV_0^{(j-1)}$.
	\EndFor
	\EndProcedure
	\Ensure{Optionally, $\Sigma_{0, r}$, $\cA_0$ and $\cB_0$ can
          be adaptively updated within the MCMC after a pre-fixed number of iterations, cf. \cite{MCMC:AnMou_2006, MCMC:HST_2001}.}
	
  \end{algorithmic}
  \caption{Base level algorithm.}
  \label{algo:level_0}
\end{algorithm}

\begin{algorithm}[th]
    \textbf{Input:} A set of samples $\mathcal{W}_{\ell} = \{
    \bv_\ell^{(k)} \}_{k = 1}^{K_\ell}$ from the level--$\ell$
    reference $p_\ell^\ast(\cdot)$, the number of MCMC iterations
    $N_\ell$, a set of MCMC samples $\mathcal{V}_{\mell} = \{
    \bv_{\mell}^{(j)} \}_{j = 1}^{N_\mell}$ on level $\mell$ and an initial MCMC state $\bV_{\ell}^{(0)}$.\\
    \textbf{Output:} A LIS basis $\basis_{\ell, r}$ and a Markov chain of posterior samples $\mathcal{V}_{\ell} = \{ \bV_\ell^{(j)} \}_{j = 1}^{N_\ell}$.
  \begin{algorithmic}[1]
	\Procedure{Level--$\ell$ LIS and MCMC}{}
	\State Lift previous LIS basis, $\basis_{\ell, c} = \Theta_{\ell, c} \, \basis_{\mell, r}$.
	\State Use $\mathcal{W}_\ell$ to solve the eigenproblem
        in \eqref{eq:def_exp_Hl} to obtain the auxiliary LIS vectors $\basis_{\ell, f}$. 
	\State Estimate the empirical covariance matrix $\Sigma_{\ell,
          r}$ from the samples in $\mathcal{W}_\ell$ and define the operators $\cA_\ell$ and $\cB_\ell$ as in \eqref{eq:operators_al}--\eqref{eq:operators_bl}. \vspace{-1mm}
	\State Compute the matrices $ \cP_{\ell,\ff}^{-1}\,\cP_{\ell, fc}^{}$ and $\cP_{\ell,\ff}^{-\frac12}$ as in  \eqref{eq:PinvP}-\eqref{eq:Phalf}. 
	\For {$j = 1, \ldots, N_\ell$}
	\State Propose a candidate $\bv_\ell^\prime = (\bv_{\ell,c}^\prime, \bv_{\ell,f}^\prime)$ using Algorithm \ref{def:conditional_proposal}, which needs $\mathcal{V}_{\mell}$.
	\State Compute the acceptance probability $\alpha_\ell^{\rm
          ML}(\bV_\ell^{(j -1 )},\bv_\ell^\prime)$ defined in
        Corollary~\ref{coro:conditional}.
	\State With probability $\alpha_\ell^{\rm ML}(\bV_\ell^{(j -1 )},\bv_\ell^\prime)$, set $\bV_\ell^{(j)} = \bv_\ell^\prime$, otherwise set $\bV_\ell^{(j)} = \bV_\ell^{(j-1)}$.
	\EndFor
	\EndProcedure
 \Ensure{Optionally, $\Sigma_{\ell, r}$,
   $\cA_\ell$, $\cB_\ell$, and the matrices $
   \cP_{\ell,\ff}^{-1}\,\cP_{\ell, fc}^{}$ and
   $\cP_{\ell,\ff}^{-\frac12}$ can
          be adaptively updated within MCMC after a pre-fixed number of iterations.}
  \end{algorithmic}
  \caption{Level--$\ell$ algorithm.}
  \label{algo:level_l}
\end{algorithm}

In both algorithms, we need to use both the LIS basis $\basis_{\ell, r}$ and an empirical covariance matrix $\Sigma_{\ell, r}$ projected onto the LIS to define operators $\cA_\ell$ and $\cB_\ell$ in the DILI proposal. 
Computing the LIS basis needs some reference distribution $p_\ell^\ast(\cdot)$. We employ the Laplace approximation to the posterior (e.g., \cite{MCMC:MWBG_2012,MCMC:Petra_etal_2014}). This way, all the samples from $p_\ell^\ast(\cdot)$ can be generated in parallel and prior to the DILI--MCMC simulation. 
The empirical covariance $\Sigma_{\ell, r}$ can be estimated using either samples drawn from the reference distribution (before the start of MCMC) or adaptively using posterior samples generated in MCMC. The latter option is the classical adaptive MCMC method \cite{MCMC:HST_2001}. 
The adaptation of $\Sigma_{\ell, r}$ is optional in Algorithms \ref{algo:level_0} and \ref{algo:level_l}.
Similar to the adaptation of the covariance, the LIS basis can also be adaptively updated using newly generated posterior samples during MCMC simulations. The implementation details for the adaptation of the LIS can be found in Algorithm 1 of \cite{MCMC:CLM_2016}.

\subsection{Pooling strategy}\label{sec:pooling}

Finally, we present an alternative proposal strategy that fully exploits the power of multilevel MCMC but reduces the dependencies of samples on different levels for a better parallel performance. In this pooling strategy, we simulate coupled multilevel Markov chains level-by-level. 

\blue{Given a set of posterior samples $\mathcal{V}_{\mell} = \{\bv_{\mell}^{(i)} \}_{i =1}^{N_{\mell}}$ on level $\mell$ with $\bv_{\mell}^{(i)}  \sim \pi_{\mell}(\cdot | \by)$, we again generate $N_\ell$ samples on level $\ell$ using the multilevel proposal mechanism \eqref{eq:coupling_proposal}--\eqref{eq:proposal_2} with conditional DILI proposals as described in Algorithm \ref{alg:conditional_dili}. However, here the inputs to Algorithm \ref{alg:conditional_dili}, i.e., the proposals $\bv'_{\ell,c}$, are drawn uniformly at random (with replacement) from the set $\mathcal{V}_{\mell}$, in contrast to using proposals $\bv'_{\mell}  \sim \pi_{\mell}(\cdot | \by)$ from a sub-sampled Markov chain on level $\mell$, as discussed in Section \ref{sec:coupling} above. 
Thus, in this pooling strategy the empirical distribution of the samples in $\mathcal{V}_{\mell}$ is used as an approximation of $\pi_{\mell}(\cdot | \by)$. 

Due to variance reduction from level to level in the multilevel MCMC algorithm (cf.~eqn.\eqref{n_l}) and the excellent mixing of our MLDILI algorithm, the effective sample size of $\mathcal{V}_{\mell}$ will in general be significantly larger than the number of samples $N_\ell$ in the sample set $\mathcal{V}_{\ell} = \{\bv_{\ell}^{(i)} \}_{i =1}^{N_{\ell}}$ that we plan to generate at level $\ell$. Thus, after some burn-in phase the set $\mathcal{V}_{\mell}$ will contain (approximately) independent samples from the coarse level posterior which are needed in the construction of the Markov chain on level $\ell$ in Algorithm \ref{algo:level_l} (Line 7).} 

With the pooling strategy, it is possible to run multiple Markov
chains at the coarse level and form the pool using the union of coarse
level samples. It parallelises much more easily and also provides
flexibility if the user decides to run further refined levels to
improve the discretisation accuracy---one can simply reuse the pool of
previously computed samples before the refinement as the coarse level
proposal. Despite the practical usefulness, we note that the formal
proof of convergence of the pooling strategy remains unclear and will
need to be addressed in future research.  


\section{Numerical experiments}\label{sec_num_exp}

In this section, the algorithms are tested on a model problem involving an elliptic PDE with random coefficients described in section \ref{sec:setup}. Numerical comparisons are then given in section~\ref{sec:comparisons}. 

\subsection{Setup}\label{sec:setup}
We consider an elliptic PDE in a domain $\Omega=[0,1]^2$ with boundary $\partial \Omega$, which models, e.g., the pressure distribution $p(\bx)$ of a stationary fluid in a porous medium described by a spatially heterogeneous permeability field $k(\bx)$. Here, $\bx\in \Omega$ denotes the spatial coordinate and $\bn(\bx)$ denotes the outward normal vector along the boundary. 

The goal is to recover the permeability field from pressure observations. We assume that the permeability field follows a log--normal prior, and thus we denote the permeability field by $k(\bx) = \exp(u(\bx))$, where $u(\bx)$ is a random function equipped with a Gaussian process prior. 
In this setting, the pressure $p(\bx)$ depends implicitly on the (random) realisation of $u(\bx)$. 

For a given realisation $u(\bx)$, the pressure satisfies the elliptic PDE
\begin{equation}\label{eq:PDE}
-\nabla\cdot\left(e^{u(\bx)}\nabla p(\bx)\right) = 0, \quad \bx\in \Omega.
\end{equation}
On the left and right boundaries, we specify Dirichlet boundary conditions, while on the top and bottom we assume homogeneous Neumann  boundary conditions:
\begin{equation}\label{eq:PDE_bc}
\begin{cases}
~~p(\bx) = 0, &\ \ \text{for} \ \ \bx\in\partial \Omega_{\text{left}}\,, \\
~~p(\bx) = 1, &\ \ \text{for} \ \ \bx\in\partial \Omega_{\text{right}} \ \ \text{and} \\
~~e^{u(\bx)} \nabla p(\bx)\cdot \bn(\bx) = 0, &\ \ \text{for} \ \ \bx\in\{\partial \Omega_{\text{top}}, \partial \Omega_{\text{bottom}}\}.
\end{cases}
\end{equation}
As the quantity of interest, we define the outflow through the left vertical boundary, i.e.
\begin{equation}\label{eq:PDE_qoi}
Q(u) = -\int_0^1 e^{u(\bx)}\frac{\partial p(\bx)}{\partial x_1}\Big|_{x_1=0}\,dx_2\, = - \int_\Omega e^{u(\bx)} \nabla p(\bx)\cdot\nabla \varphi(\bx)\,d\bx\,,
\end{equation} 
where $\varphi(\bx)$ is a linear function taking value one on $\partial \Omega_{\text{left}}$ and zero on $\partial \Omega_{\text{right}}$, as suggested in \cite{Sto:teckentrup2013}.

The Gaussian process prior for $u(\bx)$ is defined by the exponential kernel $k(\bx, \bx') = \exp(-5|\bx - \bx'|)$. Figure \ref{fig:fields} (left) displays the true (synthetic) permeability field in $\log_{10}$ scale. Noisy observations of the pressure field are collected from 71 sensors located as in Figure \ref{fig:fields} (right), with a signal-to-noise ratio 50. A likelihood function can then be defined as in \eqref{eq:like}, which, together with the prior, characterises the posterior distribution in \eqref{eq:post}.

\begin{figure}[t]
	\centering
	\includegraphics[scale = 0.38]{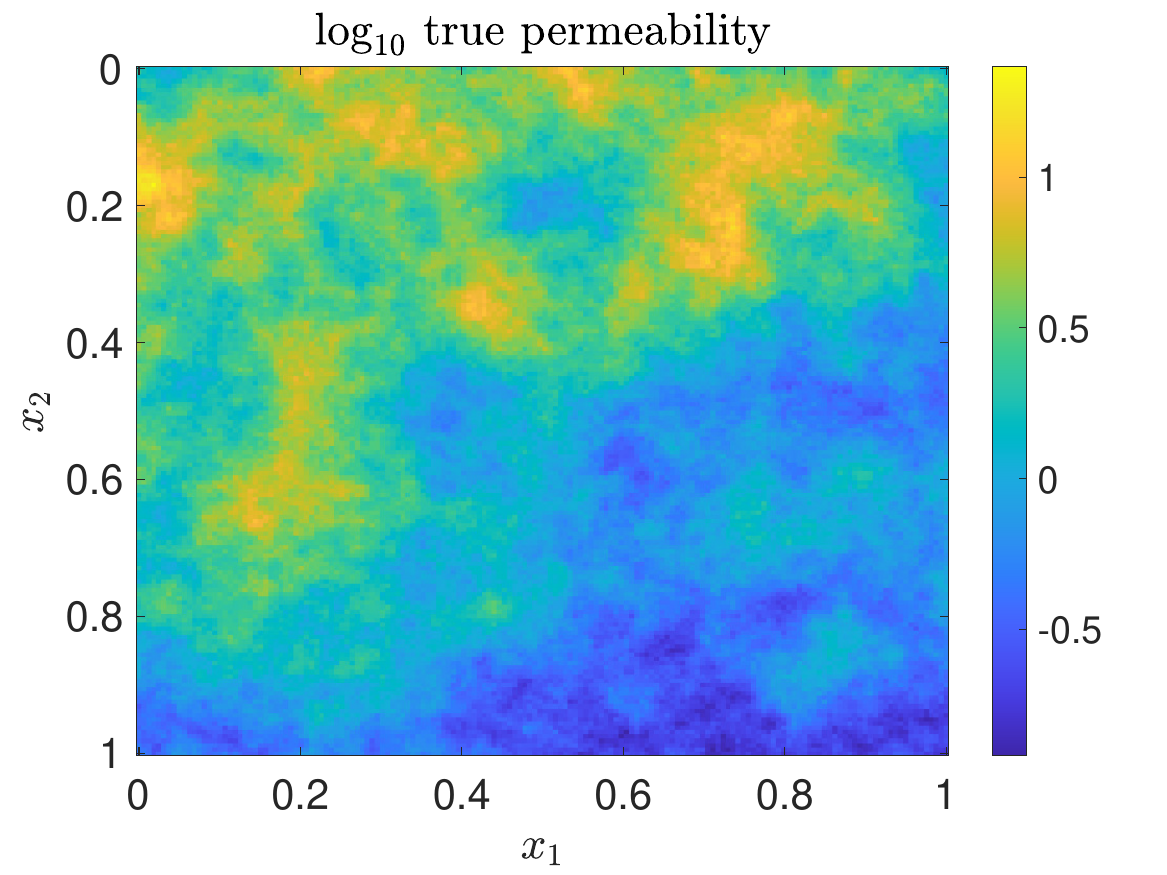}\includegraphics[scale = 0.38]{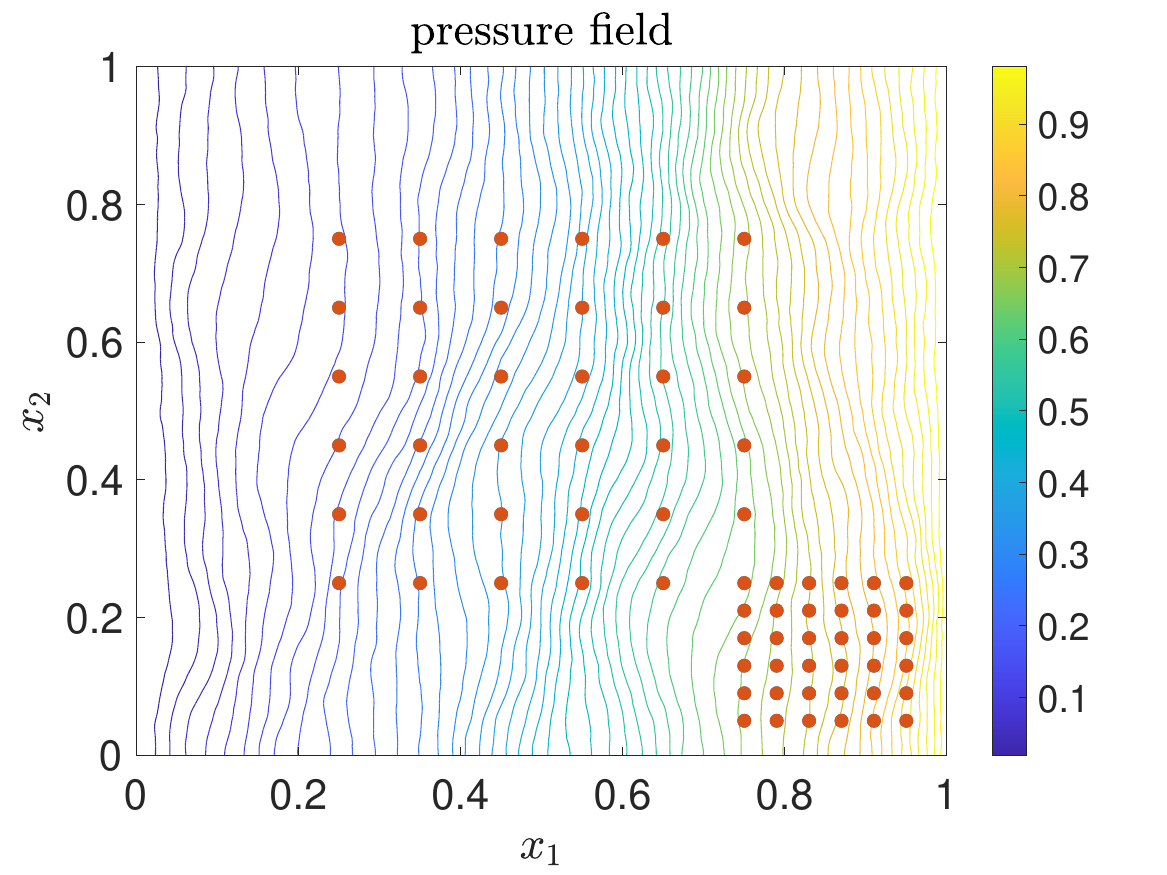}\\
	\includegraphics[scale = 0.38]{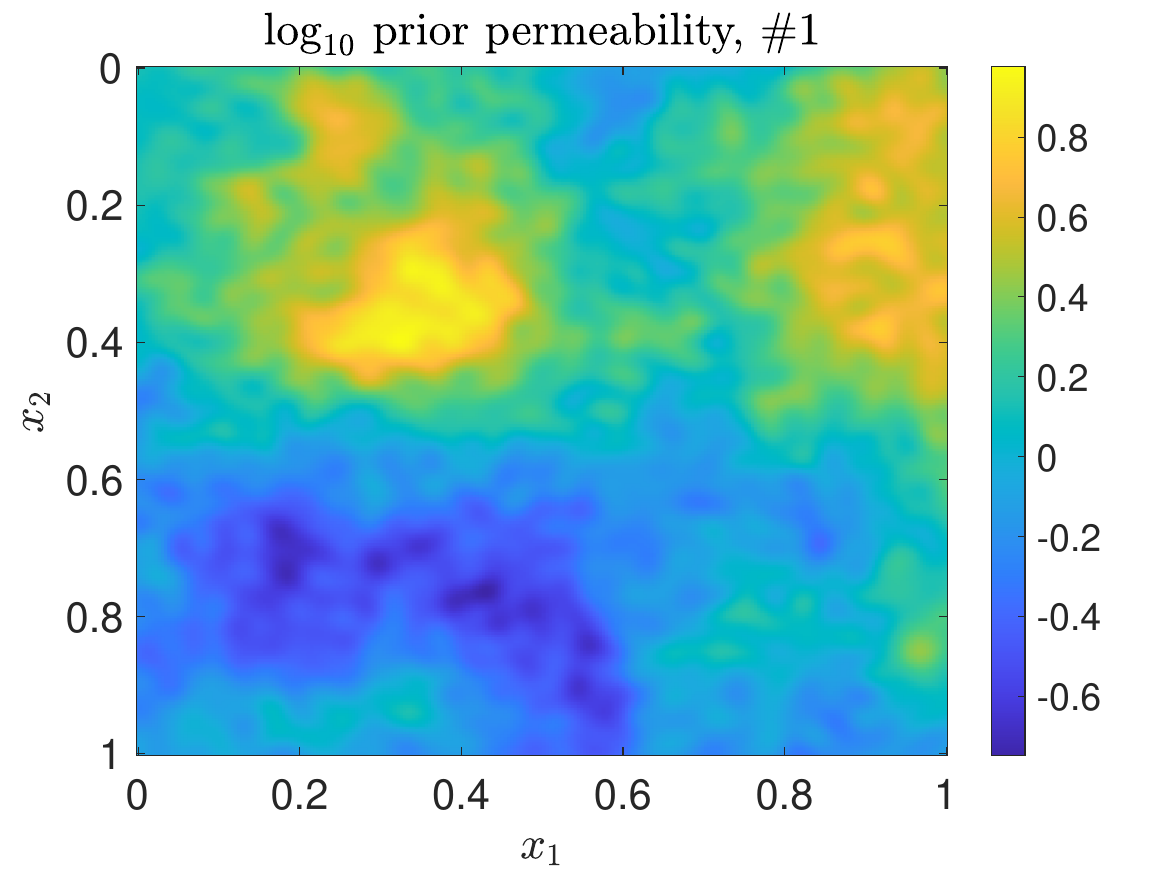}\includegraphics[scale = 0.38]{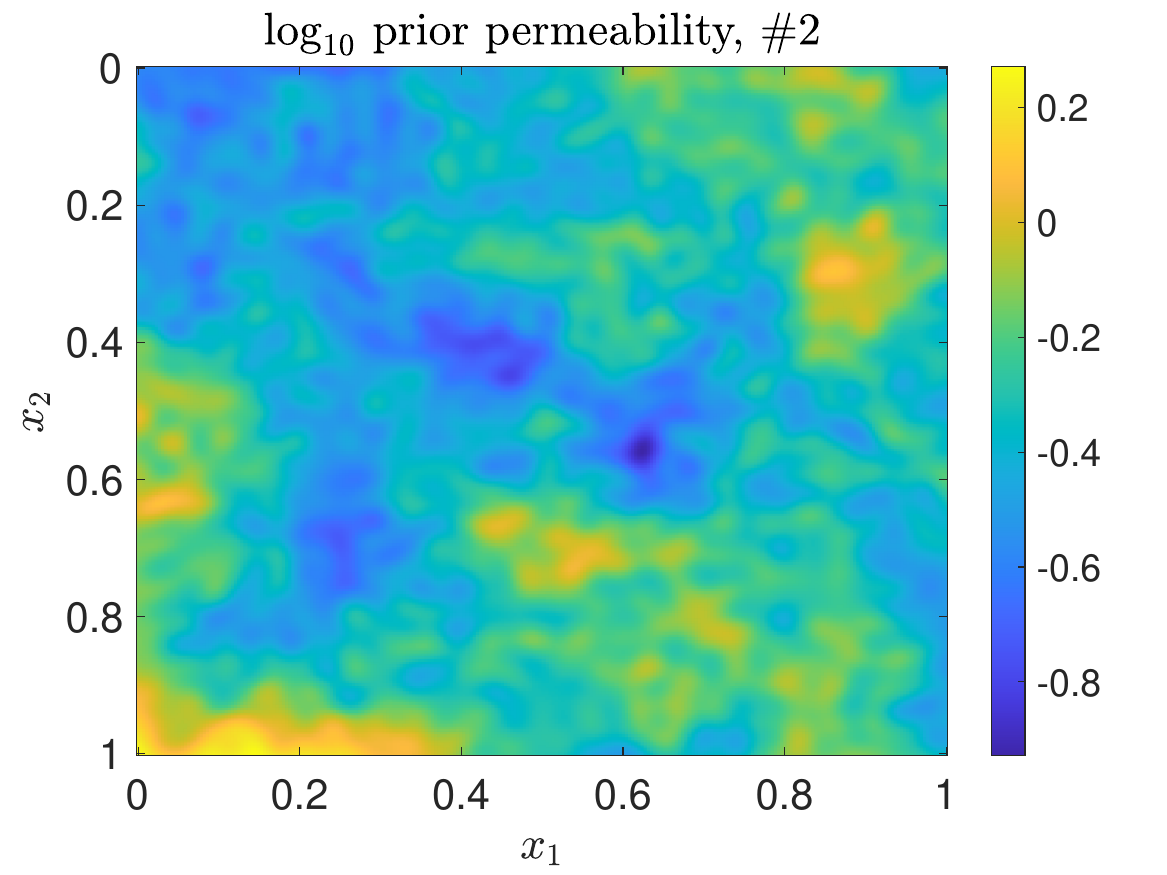}\\
	\caption{Setup of elliptic inverse problem. Top left: ``true'' permeability field used for generating the synthetic data set. Top right: observation sensors (red dots) and pressure field corresponding to ``true'' permeability field. \blue{Bottom row: Realisations of the permeability drawn from the prior.}\label{fig:fields}}
\end{figure}

In practice, \eqref{eq:PDE}--\eqref{eq:PDE_qoi} has to be solved numerically. We use standard, piecewise bilinear finite elements (FEs) on a hierarchy of nested Cartesian grids with mesh size $h_\ell = \tfrac{1}{20} \times 2^{-\ell}$, for $\ell = 0,1,2,3$. Furthermore, we approximate the unknown function $u(\bx)$ by truncated Karhunen-Lo{\`e}ve expansions with $R_\ell = 50  + 100 \times 2^\ell$ random modes, respectively. 

\subsection{Comparisons} \label{sec:comparisons}
Let us now test and compare our algorithms on the model problem described above.
First, we proceed as in section \ref{sec_MLLIS} to build a LIS at every level, using both the non-recursive and recursive constructions. Table \ref{table:LISbases} summarises the number of basis functions obtained in each case with truncation threshold $\rho=10^{-2}$, as well as the storage reduction factor given by the recursive procedure at each level. 

Because the recursive LIS construction recycles LIS bases from previous levels and enriches them with a number of auxiliary LIS vectors on each level, it is expected that the total number of basis functions obtained by the enriching procedure at each level is slightly higher than the direct (spectral) LIS on the same level.
However, in the recursive construction, the dimension of the auxiliary 
set of vectors is expected to decrease as the level increases, requiring less storage and less computational effort on finer levels, since the posterior distributions were assumed to converge with $\ell \to \infty$. 
For problems with parametrisations where the parameter dimension increases more rapidly with the discretisation level---e.g., using the same FE grid to discretise the prior covariance, the setting used in the original DILI paper~\cite{MCMC:CLM_2016}---we  expect the reduction factor to be even smaller. 

\begin{table}[t]
	\begin{center}
	\begin{tabular}{l|l l lll} 
		\toprule
		Level & 0 & 1 & 2 & 3\\
		\midrule
		Non-recursive  & 80 & 91 & 97 & 100 \\
		Recursive (added on level $\ell$)  & 80 & 21 & 19 & 12\\ 
		Recursive (total)   & 80 & 101 & 120 & 132\\ 
                \midrule
		Storage reduction factor & 1 & 0.74 & 0.60 & 0.43\\
		\bottomrule
	\end{tabular}
	\caption{LIS dimensions: Results of non-recursive construction (single-level LIS for each $\ell$) reported in first row; for the recursive construction, the number of vectors added on the current level and the total LIS dimension are given in the second and third row, respectively; the fourth row displays the storage reduction factor for the recursive procedure at each level.\label{table:LISbases}}
	
	\end{center}
\end{table}

In the comparison of sampling performances, we denote by {\bf MLpCN} the MLMCMC algorithm using the pCN proposal for the additional parameters on each level (as in \cite{MCMC:KST_2013}). The MLMCMC algorithm using the recursive LIS and the coupled DILI proposals, as summarised in Algorithms \ref{algo:level_0} and \ref{algo:level_l}, is denoted by {\bf MLDILI}. 
The integrated autocorrelation times of Markov chains constructed by MLpCN and MLDILI are reported in Table \ref{tab:IACTs}.
The IACTs for two functionals are reported for each algorithm. In the ``refined parameters'' case, at every level $\ell$ we report the average IACTs of the refined parameters $\bv_{\ell,f}$. This quantifies how well the algorithm performs in exploring the posterior distribution. In the second case, we consider the IACT of the level-$\ell$ corrections of the quantity of interest $D_\ell = Q_\ell(\rvl) - Q_{\mell}(\rvml)$.
\begin{table}[t]
	\begin{center}
	\begin{tabular}{l|ll|ll}
		& \multicolumn{2}{c|}{Refined parameters} & \multicolumn{2}{c}{$D_\ell$} \\ \hline
		{\scriptsize Level} & {\scriptsize MLDILI}                                        & {\scriptsize MLpCN}                     & {\scriptsize MLDILI}                       & {\scriptsize MLpCN}    \\ \hline
		0 & \multicolumn{1}{l|}{34}                       & 4300                       & \multicolumn{1}{l|}{9.0}     & 4100     \\
		1 & \multicolumn{1}{l|}{11}                       & 45                         & \multicolumn{1}{l|}{4.6}     & 4.9      \\
		2 & \multicolumn{1}{l|}{3.6}                      & 48                         & \multicolumn{1}{l|}{2.4}     & 2.8      \\
		3 & \multicolumn{1}{l|}{2.0}                      & 24                         & \multicolumn{1}{l|}{1.8}     & 1.9     
	\end{tabular}
\caption{Comparison of IACTs of Markov chains generated by MLDILI and MLpCN. This table reports the IACTs of the refined parameters and the level-$\ell$ correction of the quantity of interest $D_\ell = Q_\ell(\rvl) - Q_{\mell}(\rvml)$.}
\label{tab:IACTs}
\end{center}
\end{table}

In the ``refined parameters'' case, we observe a significant improvement for MLDILI over MLpCN: the coupled DILI proposal is able to reduce the IACT at every level compared to that obtained by MLpCN. 
At the base level, DILI is able to reduce the IACT by two orders of magnitude compared to that of pCN. This suggests that coarse parameter modes are very informed by the data, and thus utilising the DILI proposal is highly beneficial. 
In the case of the quantity of interest, we observe an even more impressive improvement at the base level (a factor of $456$!), while the IACTs of MLDILI and MLpCN on the finer levels are comparable. 
This suggests that the posterior distribution of the chosen quantity of interest (the integrated flux over the boundary) is not affected strongly by the high frequency parameter modes on the finer levels.
Nevertheless, in both cases, using DILI provides a huge acceleration compared to pCN.
Figure \ref{fig:iacts} compares the integrated autocorrelation times of DILI and pCN on level 0, for both the first parameter component and the quantity of interest.

\begin{figure}[t]
	\centering
	\includegraphics[scale = 0.48]{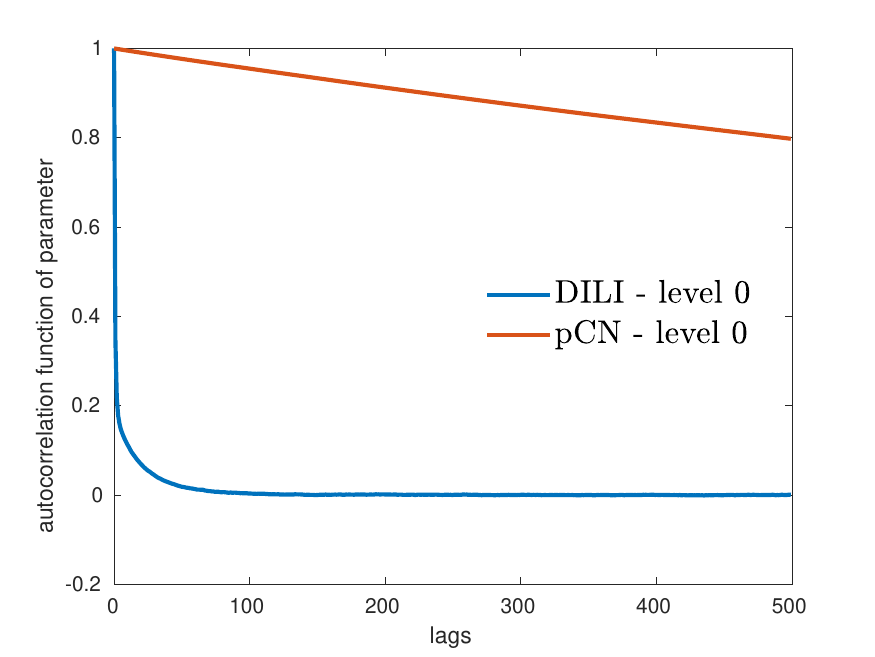}\hspace*{-0.15cm}\includegraphics[scale = 0.48]{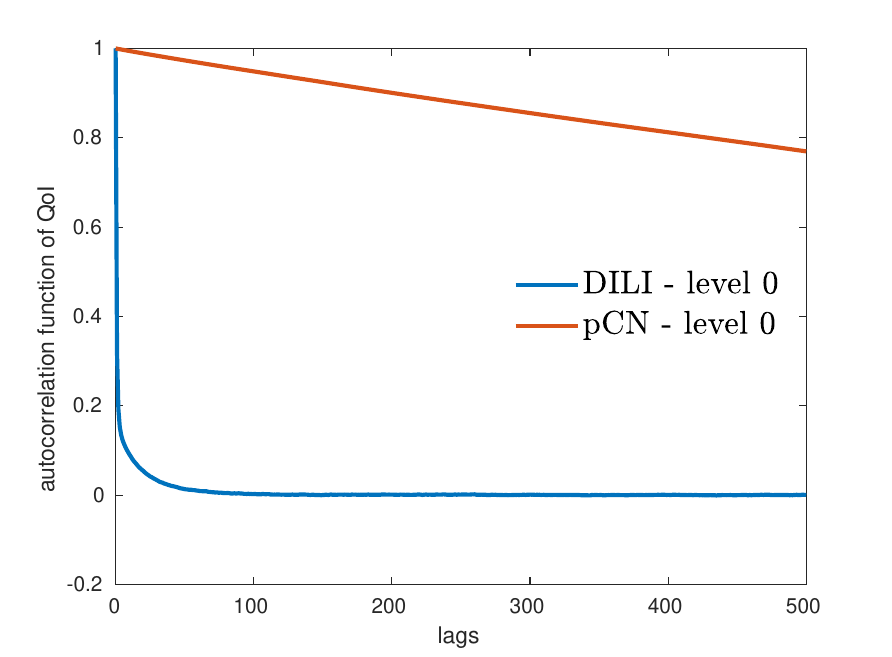}
	
	\caption{Autocorrelation functions of the chains $\{ \big(\bV^{(j)}_{0}\big)_1 \}$ and $\{Q_0(\bV^{(j)}_0)\}$ on the coarsest level (Blue: DILI. Red: pCN). \label{fig:iacts}}
\end{figure}

The IACTs for the level-$\ell$ corrections of the quantity of interest in Table \ref{tab:IACTs} suggest that using a mixed strategy---in which one employs the LIS and DILI only at the coarsest level and uses pCN in refined levels---is also a reasonable approach in cases where the important likelihood-informed directions that have any influence on the quantity of interest are already well enough identified in the  base-level LIS.
We refer to this as the {\bf MLmixed} strategy. 

We compare the computational performance of the three multilevel algorithms (MLDILI, MLpCN, MLmixed) with the two single level algorithms using DILI and pCN proposals.
The finite element model and all MCMC algorithms are implemented in MATLAB; we use sparse Cholesky factorisation \cite{Lin:Cholmod} to solve the finite element systems and ARPACK \cite{Lin:ARPACK} to solve the eigenproblems.    
All simulations are carried out on a workstation equipped with 28 cores (two Intel Xeon E5-2680 CPUs).
The performance of MLmixed is only estimated using the IACTs and the actual computing times measured in the MLDILI and MLpCN runs.

The computational complexities of the five algorithms for approximating $\Ev_{\pi}[Q]$ on (discretisation) levels $L = 1, 2$ and $3$ with $Q$ defined in \eqref{eq:PDE_qoi} are compared in Figure \ref{fig:varred_accsamples} (right). In the multilevel estimators, the coarsest level is always $\ell=0$, so that the number of levels is $2,3$ and $4$, respectively.
The sampling error tolerance on each level is adapted to the corresponding bias error due to finite element discretisation and parameter truncation, such that the squared bias is equal to the variance of the estimator. The bias errors were estimated beforehand to be $9\times 10^{-3}$, $4\times 10^{-3}$, and $2\times 10^{-3}$ on levels $L = 1, 2$ and $3$, leading to a total error of $1.27\times 10^{-2}$, $5.7\times 10^{-3}$, $2.8\times 10^{-3}$, respectively. Those bias estimates are plotted in Figure \ref{fig:varred_accsamples} (left) together with estimates of $\Var_{\pi_\ell}(Q_\ell)$ and $\Var_{\Delta_{\ell, \mell}}( Q_\ell - Q_{\mell} )$, which suggest that $\theta_b \approx 0.5$ and $\theta_v \approx 0.5$ in Assumptions \ref{assum_bias}(i) and \ref{assum_within}. This agrees with the theoretical results in \cite{MCMC:KST_2013}. The cost per sample is dominated by the sparse Cholesky factorisation on each level and scales roughly like $\mathcal{O}(M_\ell^{1.2})$, so that $\theta_c \approx 1.2$ in Assumption \ref{assum_bias}(ii). Optimally scaling multigrid solvers exist for this model problem, but for the FE problem sizes considered here they are more costly in absolute terms. Moreover, we can also exploit the fact that the adjoint problem is identical to the forward problem here, so that the Cholesky factors can be reused for the adjoint solves required in the LIS construction.

Let us now discuss the results. Single level pCN becomes impractical in this example, since the data is very  informative and leads to an extremely low effective sample size. Some of this bad statistical efficiency is inherited by MLpCN, at least in absolute terms, due to the poor effective sample size on level $0$. Asymptotically this effect disappears and the rate of growth of the cost is smallest for MLpCN with an observed assymptotic cost of about $\mathcal{O}(\epsilon^{-2.3})$. As observed in \cite{MCMC:KST_2013}, this is better than the theoretically predicted asymptotic rate and is likely a pre-asymptotic effect due to the high cost on level 0. 
Unsurprisingly, given the low IACTs reported in Table~\ref{tab:IACTs}, the methods based on DILI proposals all perform significantly better. MLDILI and MLmixed perform almost identically, since the corresponding IACTs on all levels are very similar. They are consistently better than single-level DILI and the asymptotic rate of growth of the cost is also better, $\mathcal{O}(\epsilon^{-3.4})$ versus $\mathcal{O}(\epsilon^{-4.1})$. Both rates are consistent with the theoretically predicted rates in Theorem~\ref{thm:mlmcmc}, given the estimates for $\theta_b, \theta_v, \theta_c$ above. For the highest accuracies, MLDILI is almost 4 times faster than DILI, and due to the better asymptotic behaviour this reduction factor will grow as $\varepsilon \to 0$. For grid level $L=4$, even MLpCN is expected to outperform single-level DILI, but the computational costs of the estimators for higher accuracies are starting to become impractical even using the multilevel acceleration, as the dashed line representing one CPU day in Figure \ref{fig:varred_accsamples} (right) indicates. 

\begin{figure}[t]
	\centering
	\includegraphics[scale = 1]{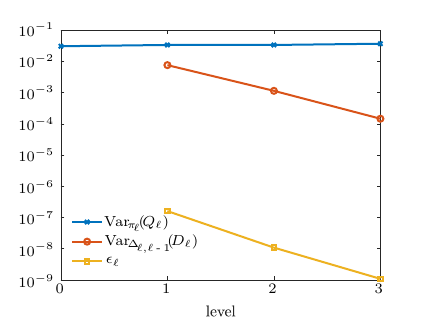}
	\includegraphics[scale = 1]{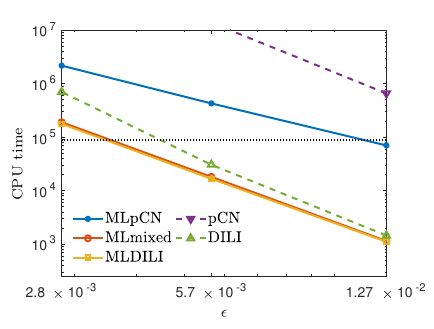}
	\caption{Left: the variance $\Var_{\pi_\ell}(Q_\ell)$ (blue) and the bias $\epsilon_\ell = \big| \Ev_{\mu_y}\big[Q\big] - \Ev_{\pi_\ell}\big[Q_\ell\big] \big|$ (yellow) at each level, and the cross-level variances $\Var_{\Delta_{\ell, \mell}}( D_{\mell} ) = \Var_{\Delta_{\ell, \mell}}( Q_\ell - Q_{\mell} )$ (red) used for estimating the CPU time for various MCMC methods. Right: Total CPU time (in seconds) for various methods to achieve different total error tolerances. The LISs are constructed by recycling Cholesky factors. The dotted line represents a CPU day.\label{fig:varred_accsamples}}
\end{figure}

The dominating cost in solving the eigenproblems \eqref{eq:exp_H0} and \eqref{eq:def_exp_Hl} is the Cholesky factorisation. As mentioned above, sparse direct solvers are used to solve the stationary forward model and we are able to recycle the Cholesky factors from the forward solve to compute the actions of the adjoint model in \eqref{eq:exp_H0} and \eqref{eq:def_exp_Hl} for each sample. 
As a result, the computational cost of building the LIS is negligible compared to that of the MCMC simulation here (for both the single level and the recursive construction).
This also explains why MLmixed performs almost identically to MLDILI.

\begin{figure}[t]
	\centering
	\includegraphics[scale = 1]{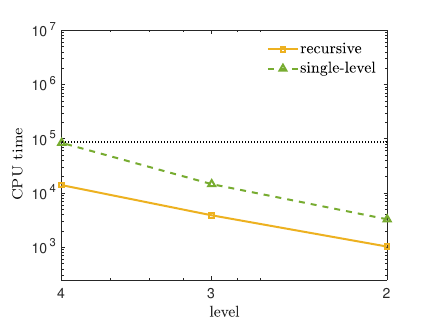}
	\includegraphics[scale = 1]{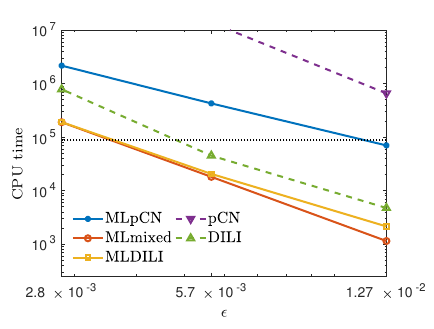}
	\caption{Left: Total CPU time (in seconds) for the single level and recursive constructions of the LISs at level $2, 3$ and $4$. Right: Total CPU time (in seconds) for various methods to achieve different error tolerances. The LISs are constructed without recycling Cholesky factors. The dotted line represents a CPU day.\label{fig:varred_accsamples_new}}
\end{figure}

However, in many other applications this is not possible due to the high storage cost or when the adjoint is different. Each action of the adjoint problem typically has a comparable cost to solving the forward model in the stationary case. It can even be more expensive than solving the forward model in time-dependent problems. 
To provide a thorough comparison in that case, we also report the total CPU time of all the estimators in Figure \ref{fig:varred_accsamples_new} when the LIS setup cost is included. Here, we compute both the single level LIS and the recursive LIS without storing the Cholesky factors, to mimic the behaviour in the general, large-scale case.
In this setup, we observe that a significant amount of computing effort is spent on building the LIS, and thus MLmixed and MLDILI significantly outperform the single level DILI for all error thresholds. MLmixed is more than 4 times faster than DILI even for the largest error threshold of $1.27 \times 10^{-2}$. The construction of the single-level LIS requires two times more CPU time than performing the actual MCMC simulation in that case. 
In comparison, a significant number of adjoint model solves can be saved by the recursive LIS construction. 
Furthermore, we do expect that the computational cost for constructing the recursive LIS will stop increasing, since the dimension of the auxiliary LIS will eventually be zero at higher levels.
Overall, for large--scale problems where the adjoint cannot be cheaply computed by recycling the forward model simulation, the recursive LIS construction, and hence the MLDILI, is clearly more computationally efficient than the single level DILI.


\section{Conclusion}\label{sec_conclusion}

We integrate the dimension-independent likelihood-informed MCMC from \cite{MCMC:CLM_2016} into the multilevel MCMC framework in
\cite{MCMC:KST_2013} to improve the computational efficiency of
estimating the expectation of functionals of interests over posterior measures.
Several novel elements are introduced in this integration. We first
design a Rayleigh-Ritz procedure to recursively construct likelihood
informed subspaces that exploit the hierarchy of model
discretisations. 
The resulting hierarchical LIS needs lower computational effort to
construct and has lower operation cost compared to the original LIS 
proposed in \cite{DimRedu:Cui_etal_2014}.
Then, we present a new pooling strategy to couple Markov chains on
consecutive levels. This enables more flexible parallelisation and
management of computing resources. 
Finally, we design new coupled DILI proposals by exploiting the
hierarchical LIS, so that the DILI proposal can be applied in the
multilevel  MCMC setting. 
We also demonstrate the efficacy of our integrated approach on a model
inverse problem governed by an elliptic PDE.

\section*{Data availability statement}

No new data were created or analysed in this study.

\section*{Acknowledgements}
TC acknowledges support from the Australian Research Council under the
grant DP210103092. GD was supported by the EPSRC Centre for Doctoral
Training in Statistical Applied Mathematics at Bath (EP/L015684/1). RS
acknowledges support by the Deutsche Forschungsgemeinschaft (German
Research Foundation) under Germany’s Excellence Strategy EXC 2181/1 --
390900948 (the Heidelberg STRUCTURES
Excellence Cluster).

{

}

\bigskip\bigskip


\renewcommand{\appendix}{\par
  \setcounter{section}{0}
  \setcounter{subsection}{0}
  \gdef\thesection{\Alph{section}}
}

\appendix

\section{\blue{Computational complexity of hierarchical LIS}}\label{sec:LIS_cost}
Here we develop heuristics---under the following set of restrictive assumptions---to compare the complexities of the construction of the hierarchical LIS and of the single-level LIS, constructed directly on level $L$. 
\begin{assumption}\label{assum:LIS}
\;
\begin{enumerate}
	\item \!The parameter dimensions satisfy $R_\ell = R_0  e ^ {\beta_{\rm p} \ell}$ for some  $\beta_{\rm p} > 0$.\!
	\item The number of auxiliary LIS basis vectors satisfies $s_\ell \leq s_0 \, e ^ { - \beta_{\rm r} \ell}$ for some $\beta_{\rm r} > 0$. 
  \item The degrees of freedom in the forward model satisfy $M_\ell
    = M_0 \, e ^ {\beta_{\rm m} \ell}$ for some  $\beta_{\rm m} > 0$. 
\item The computational cost of a matrix vector product with one sample of the Gauss-Newton Hessian
  $\hessian_\ell(\bv_\ell^{(k)})$ is proportional to one evaluation of
  the forward model and thus $\cO(M_\ell^{\vartheta_{\rm c}})$
  (cf.~Assumption \ref{assum_bias}).
  \item The number of samples to compute the sample-averaged
    Gauss-Newton Hessian is the same on all levels, i.e., $K_\ell = K$ independent of $\ell$.
	\item For the single-level LIS constructed on level $L$, we
          assume that the LIS dimension satisfies
          $r^{\rm single}_L \geq c\,r_0 $ for some constant $c > 0$. 
\end{enumerate}
\end{assumption}

The storage cost of the hierarchical LIS basis and the storage cost of
the single-level LIS basis on level $L$ are, respectively, 
\[
\zeta_{\rm multi} = {\textstyle \sum_{l = 0}^{L}}\, R_\ell\,s_\ell, \quad \textrm{and} \quad \zeta_{\rm single} = \, R_L\, r^{\rm single}_L\,.
\]
The floating point operations for one matrix vector product with the hierarchical LIS
basis and with the single-level LIS basis are $O\big( \zeta_{\rm
  multi} \big)$ and $O\big( \zeta_{\rm single} \big)$, respectively,
with the same hidden constant.

\begin{corollary}\label{coro:speedup}
The reduction factor of storing and operating with the hierarchical
LIS basis (as opposed to the standard single-level LIS on level $L$) satisfies the upper bound 
\begin{equation}
\frac{\zeta_{\rm multi}}{\zeta_{\rm single} } \leq \frac{1}{c}\,
\min\Big( L+1, \frac{1}{1 - e^{ - | \beta_{\rm p} - \beta_{\rm r} |}
}\Big) \,e^{- \min(\beta_{\rm p}, \beta_{\rm r}) L} \,.
\end{equation}
\end{corollary}
\begin{proof}

Using Assumption \ref{assum:LIS}, the required storage for the
hierarchical and for the single-level LIS bases can be bounded by
\begin{align*}
\zeta_{\rm multi}  = {\textstyle \sum_{l = 0}^{L}}\, R_\ell\,s_\ell
                    \leq  R_0 \, s_0 {\textstyle \sum_{l = 0}^{L}} \,
                    e^{ (\beta_{\rm p} - \beta_{\rm r}) \ell} \ \ \
                    \text{and} \ \ \ 
\zeta_{\rm single}  = \, R_L\, r_L \geq  c\, R_0 \, s_0\, e ^ {\beta_{\rm p} L}.
\end{align*}
Thus, the reduction factor satisfies
\begin{equation}
\label{app:eq1}
\frac{\zeta_{\rm multi}}{\zeta_{\rm single} } \leq \frac{1}{c} \, e^{ -\beta_{\rm p} L}\Big({\textstyle \sum_{l = 0}^{L}} \, e ^ { (\beta_{\rm p} - \beta_{\rm r}) \ell}\Big).
\end{equation}
We first consider the case $\beta_{\rm p} \neq \beta_{\rm r}$. Using the property of geometric series, we have 
\[
{\textstyle \sum_{l = 0}^{L}} \, e^{ (\beta_{\rm p} - \beta_{\rm r}) \ell} = \frac{1 - e^{ (\beta_{\rm p} - \beta_{\rm r}) (L+1)}  }{1 - e^{(\beta_{\rm p} - \beta_{\rm r} ) }}.
\]
For the case $\beta_{\rm p} < \beta_{\rm r}$, the reduction factor satisfies
\begin{equation}
\frac{\zeta_{\rm multi}}{\zeta_{\rm single} }
\leq \frac{1}{c} \, e^{ -\beta_{\rm p} L} \frac {1 - e^{ (\beta_{\rm p} - \beta_{\rm r}) (L+1)}  }{1 - e^{(\beta_{\rm p} - \beta_{\rm r} ) }},
%
%
\end{equation}
whereas for $\beta_{\rm p} > \beta_{\rm r}$, the reduction factor satisfies
\begin{equation}
\frac{\zeta_{\rm multi}}{\zeta_{\rm single} }
\leq \frac{1}{c} \, e^{ - \beta_{\rm p} L} \frac{1 - e^{ (\beta_{\rm p} - \beta_{\rm r}) (L+1)}  }{1 - e^{(\beta_{\rm p} - \beta_{\rm r} ) }} 
= \frac{1}{c} \, e^{ -\beta_{\rm r} L} \frac{1 - e^{ (\beta_{\rm r} - \beta_{\rm p}) (L+1)}  } {1 - e^{(\beta_{\rm r} - \beta_{\rm p} ) }}.
\end{equation}
In both cases, the reduction factor can be expressed as
\begin{equation}\label{eq:appen_2nd}
\frac{\zeta_{\rm multi}}{\zeta_{\rm single} }
\leq \frac{1}{c} \, e^{ -\min(\beta_{\rm p}, \beta_{\rm r}) L}  \frac {1 - a^{L+1}  }{1 - a},
\end{equation}
where $a = e^{ - | \beta_{\rm p} - \beta_{\rm r} |} \in (0, 1)$. 
Using induction, one can easily show that
\begin{equation}
\frac{1 - a^{L+1}  }{1 - a} \leq \min\Big( L+1, \frac{1}{ 1 - a} \Big) , \quad \forall L \geq 0, \forall a \in (0, 1),
\end{equation}
which completes the proof for $\beta_{\rm p} \neq \beta_{\rm r}$.

\blue{For $\beta_{\rm p} = \beta_{\rm r} = \min(\beta_{\rm p}, \beta_{\rm
  r})$ the result of Corollary \ref{coro:speedup} follows directly 
  from \eqref{app:eq1} since in that case $ \sum_{l = 0}^{L} \, e ^ { (\beta_{\rm p} - \beta_{\rm r}) \ell } = L+1$.}
\end{proof}

Using a similar derivation, we can also obtain the reduction factor for constructing the hierarchical LIS basis. 
The number of matrix vector products (with the sample-averaged Gauss-Newton Hessian $\widehat\hessian_0$) in the construction of the base level LIS via the eigenproblems \eqref{eq:exp_H0} is linear in the number of leading eigenvectors obtained, i.e., $\cO(s_0)$. 
The same holds for the number of matrix vector products with $\widehat\hessian_\ell$ in
the construction of the auxiliary LIS vectors in the recursive
enrichment solving the eigenproblems in \eqref{eq:def_exp_Hl}. 
Thus, the overall computational complexities for constructing the hierarchical LIS basis is
\[
\chi_{\rm multi} =  \cO\big( {\textstyle K \sum_{l = 0}^{L}} \, s_\ell \, M_\ell^{\vartheta_{\rm c}} \big).
\]
Similarly, the construction of the single level LIS on level $L$ is
\[
\chi_{\rm single} =  \cO\big(K r^{\rm single}_L \, M_L^{\vartheta_{\rm c}} \big),
\]
where the prefactors are the same.
The following corollary can be proved in the same way as Corollary
\ref{coro:speedup}, since we have assumed that 
$M_\ell^{\vartheta_{\rm c}} = M_0^ {\vartheta_{\rm c}} \, e ^ {\beta_{\rm m} \vartheta_{\rm c} \ell}$.
\begin{corollary}
The reduction factor of building the hierarchical LIS basis (as
opposed to the standard single-level LIS basis on level $L$) satisfies the upper bound
\begin{equation}
\frac{\chi_{\rm multi}}{\chi_{\rm single} } \leq \frac{1}{c}\,
\min\Big( L+1, \frac{1}{1 - e^{ - | \beta_{\rm m}\vartheta_{\rm c} -
    \beta_{\rm r} |} }\Big) \, e^{ - \min(\beta_{\rm m}\vartheta_{\rm c} \,,\, \beta_{\rm r}) L}\,.
\end{equation}
\end{corollary}


\section{\blue{Proof of Corollary \ref{coro:conditional}}}\label{sec:proof_conditional}
Due to the acceptance probability \eqref{eq:accept_multi}, we have 
\begin{align*}
\beta^{}_\ell(\bv_{\ell}^\ast, \bv_{\ell}'  ) & = \min\bigg\{1, \frac{\pi^{}_\ell\big(\bv_{\ell}' \,|\, \data\big)\,\pi^{}_{\mell}\big(\bv_{\mell}^\ast\,|\, \data\big)}{\pi^{}_\ell\big(\bv_{\ell}^\ast \,|\, \data\big)\,\pi^{}_{\mell}\big(\bv_{\mell}^\prime\,|\, \data\big)}  \,
 \frac{q \big(  \bv_{\ell, f}^\ast \,|\, \bv_{\ell}^\prime, \bv_{\mell}^\ast \big)}{q \big( \bv_{\ell,f}^\prime \,|\, \bv_{\ell}^\ast,\bv_{\mell}^\prime \big)} \bigg\},
\end{align*}
where, by definition,
\[
\frac{\pi^{}_\ell\big(\bv_{\ell}' \,|\, \data\big)\,\pi^{}_{\mell}\big(\bv_{\mell}^\ast\,|\, \data\big)}{\pi^{}_\ell\big(\bv_{\ell}^\ast \,|\, \data\big)\,\pi^{}_{\mell}\big(\bv_{\mell}^\prime\,|\, \data\big)} 
= \frac{p^{}_\ell\big(\bv_{\ell}'\big)\,p^{}_{\mell}\big(\bv_{\mell}^\ast\big)}{p^{}_\ell\big(\bv_{\ell}^\ast\big)\,p^{}_{\mell}\big(\bv_{\mell}^\prime\big)} \, \frac{\exp\big( \!\!-\! \potential^{}_\ell \big(\bv_\ell'; \data\big) \!+\! \potential^{}_{\mell} \big(\bv_{\mell}'; \data\big) \big)}{\exp\big( \!\!-\! \potential^{}_\ell \big(\bv_\ell^\ast; \data\big) \!+\! \potential^{}_{\mell} \big(\bv_{\mell}^\ast; \data\big) \big)}\,,
\]
such that we can write
\begin{equation}
\label{app:ratio}
\beta^{}_\ell(\bv_{\ell}^\ast, \bv_{\ell}' ) \!=\! 
\min\!\bigg\{\! 1, \underbrace{ \frac{p^{}_\ell \big(\bv_{\ell}'\big)\,p^{}_{\mell} \big(\bv_{\mell}^\ast\big)\,q \big(  \bv_{\ell, f}^\ast | \bv_{\ell}^\prime, \bv_{\mell}^\ast \big)}{p^{}_\ell \big(\bv_{\ell}^\ast\big)\,p^{}_{\mell} \big(\bv_{\mell}^\prime\big)\,q \big( \bv_{\ell,f}^\prime | \bv_{\ell}^\ast,\bv_{\mell}^\prime \big)} }_{\circled{\small 1}} \!
\underbrace{ \frac{\exp\big( \!\!-\! \potential^{}_\ell \big(\bv_\ell'; \data\big) \!+\! \potential^{}_{\mell} \big(\bv_{\mell}'; \data\big) \big)}{\exp\big( \!\!-\! \potential^{}_\ell \big(\bv_\ell^\ast; \data\big) \!+\! \potential^{}_{\mell} \big(\bv_{\mell}^\ast; \data\big) \big)} }_{\circled{\small 2}} \!\!\bigg\} .
\end{equation}

The level $\ell$ parameter vectors can be split as $\bv_{\ell}' = (\bv_{\ell,f}', \bv_{\ell,c}')$ and $\bv_{\ell}^\ast = (\bv_{\ell,f}^\ast, \bv_{\ell,c}^\ast)$. 
and we have $\bv_{\ell,c}' = \bv_{\mell}'$ and $\bv_{\ell,c}^\ast =
\bv_{\mell}^\ast$ by construction in the coupling procedure. Thus, 
\begin{equation}\label{eq:ratio1}
\circled{\small 1} = \frac{p^{}_\ell \big(\bv_{\ell,f}', \bv_{\ell,c}'\big)\,p^{}_{\mell} \big(\bv_{\ell,c}^\ast\big)\,q \big(  \bv_{\ell, f}^\ast | \bv_{\ell,f}', \bv_{\ell,c}', \bv_{\ell,c}^\ast \big)}{p^{}_\ell \big(\bv_{\ell,f}^\ast, \bv_{\ell,c}^\ast\big)\,p^{}_{\mell} \big(\bv_{\ell,c}'\big)\,q \big( \bv_{\ell,f}^\prime | \bv_{\ell,f}^\ast, \bv_{\ell,c}^\ast,\bv_{\ell,c}' \big)}.
\end{equation}
The density of the conditional DILI proposal $q \big( \bv_{\ell,f}^\prime | \bv_{\ell,f}^\ast, \bv_{\ell,c}^\ast,\bv_{\ell,c}' \big)$ is defined as
\begin{equation}\label{eq:ratio_DILI_cond}
q \big( \bv_{\ell,f}^\prime | \bv_{\ell,f}^\ast, \bv_{\ell,c}^\ast,\bv_{\ell,c}' \big) = \frac{q \big( \bv_{\ell,f}', \bv_{\ell,c}' | \bv_{\ell,f}^\ast, \bv_{\ell,c}^\ast \big)}{q \big( \bv_{\ell,c}' | \bv_{\ell,f}^\ast, \bv_{\ell,c}^\ast \big)},
\end{equation}
that is the ratio between the DILI proposal density and the marginal DILI proposal density, which takes the form
\begin{equation}
q \big( \bv_{\ell,c}' | \bv_{\ell,f}^\ast, \bv_{\ell,c}^\ast \big) \equiv \int q \big( \bv_{\ell,f}', \bv_{\ell,c}' | \bv_{\ell,f}^\ast, \bv_{\ell,c}^\ast \big) d\bv_{\ell,f}'.
\end{equation}

Due to Corollary \ref{coro:accept_DILI}, the DILI proposal
$q(\bv_{\ell}' |  \bv_{\ell}^\ast)$ has the prior distribution
$p_\ell(\bv_\ell)$ as invariant measure, i.e.,
\begin{equation}
p^{}_\ell \big(\bv_{\ell}^\ast\big) q \big( \bv_{\ell}' | \bv_{\ell}^\ast \big) = p^{}_\ell \big(\bv_{\ell}' \big).
\end{equation}
Hence, if $\bv_{\ell}^\ast = (\bv_{\ell,f}^\ast, \bv_{\ell,c}^\ast)$ is drawn from the prior $p_\ell(\bv_{\ell})$, then the proposal candidate $\bv_{\ell}' = (\bv_{\ell,f}', \bv_{\ell,c}')$ also follows the prior $p_\ell(\bv_{\ell})$. 
Furthermore, if $\bv_{\ell}^\ast$ is drawn from $p_\ell(\bv_{\ell})$,
then the marginal DILI proposal $q \big( \bv_{\ell,c}' |
\bv_{\ell,f}^\ast, \bv_{\ell,c}^\ast \big)$ generates candidates with
coarse components that follow the marginal prior
\begin{equation*}
\int p^{}_\ell(\bv_{\ell,f}', \bv_{\ell,c}') d\bv_{\ell,f}',
\end{equation*}
which for our particular choice of parametrisation is
the same as the prior $p_{\mell} \big(\bv_{\ell,c}'\big)$ on level
$\mell$, that is, $p^{}_\ell \big(\bv_{\ell}^\ast\big) q \big( \bv_{\ell,c}' | \bv_{\ell}^\ast \big) = p^{}_{\mell} \big(\bv_{\ell,c}' \big)$. 
Using this identity and substituting \eqref{eq:ratio_DILI_cond} into \eqref{eq:ratio1}, the ratio \circled{\small 1} can be simplified to
\begin{align*}
\circled{\small 1} & = \frac{p^{}_\ell \big(\bv_{\ell,f}', \bv_{\ell,c}'\big)\,q \big(  \bv_{\ell, f}^\ast, \bv_{\ell,c}^\ast  | \bv_{\ell,f}', \bv_{\ell,c}'\big)\,q \big( \bv_{\ell,c}' | \bv_{\ell,f}^\ast, \bv_{\ell,c}^\ast \big)\,p^{}_{\mell} \big(\bv_{\ell,c}^\ast\big)}
{p^{}_\ell \big(\bv_{\ell,f}^\ast, \bv_{\ell,c}^\ast\big)\,q \big( \bv_{\ell,f}', \bv_{\ell,c}' | \bv_{\ell,f}^\ast, \bv_{\ell,c}^\ast \big)\,q \big( \bv_{\ell,c}^\ast | \bv_{\ell,f}', \bv_{\ell,c}' \big)\,p^{}_{\mell} \big(\bv_{\ell,c}'\big)} \\
& = \frac{p^{}_\ell \big(\bv_{\ell,f}^\ast, \bv_{\ell,c}^\ast\big)\,q \big( \bv_{\ell,c}' | \bv_{\ell,f}^\ast, \bv_{\ell,c}^\ast \big)\,p^{}_{\mell} \big(\bv_{\ell,c}^\ast\big)}
{p^{}_\ell \big( \bv_{\ell,f}', \bv_{\ell,c}' \big)\,q \big( \bv_{\ell,c}^\ast | \bv_{\ell,f}', \bv_{\ell,c}' \big)\,p^{}_{\mell} \big(\bv_{\ell,c}'\big)} 
 = \frac{p^{}_{\mell} \big(\bv_{\ell,c}' \big)\,p^{}_{\mell} \big(\bv_{\ell,c}^\ast\big)}
{p^{}_{\mell} \big(\bv_{\ell,c}^\ast \big)\,p^{}_{\mell}
  \big(\bv_{\ell,c}'\big)} = 1.
\end{align*}
The result then follows immediately from \eqref{app:ratio}.

\end{document}